\documentclass[a4paper,fleqn]{cas-dc}

\usepackage[numbers]{natbib} 

\usepackage{caption}
\usepackage{multirow}
\usepackage{longtable}
\usepackage{tabu}
\usepackage{colortbl,booktabs}
\usepackage{amsmath}
\usepackage{amsthm,amssymb}
\usepackage{amsfonts}
\usepackage{diagbox}
\usepackage{multicol}
\usepackage{url}
\usepackage{graphicx}
\usepackage{subfig}

\usepackage{verbatim}
\usepackage{algorithm}
\usepackage{algorithmicx}

\usepackage{soul,color,xcolor}

\usepackage[noend]{algpseudocode}
\makeatletter
\newenvironment{breakablealgorithm}
  {
   \begin{center}
     \refstepcounter{algorithm}
     \hrule height.8pt depth0pt \kern2pt
     \renewcommand{\caption}[2][\relax]{
       {\raggedright\textbf{\ALG@name~\thealgorithm} ##2\par}%
       \ifx\relax##1\relax 
         \addcontentsline{loa}{algorithm}{\protect\numberline{\thealgorithm}##2}%
       \else 
         \addcontentsline{loa}{algorithm}{\protect\numberline{\thealgorithm}##1}%
       \fi
       \kern2pt\hrule\kern2pt
     }
  }{
     \kern2pt\hrule height.8pt depth0pt \kern2pt\relax
   \end{center}
  }
\makeatother
\errorcontextlines\maxdimen
\makeatletter
    \newcommand*{\algrule}[1][\algorithmicindent]{\makebox[#1][l]{\hspace*{.5em}\thealgruleextra\vrule height \thealgruleheight depth \thealgruledepth}}%
\newcommand*{\thealgruleextra}{}
\newcommand*{\thealgruleheight}{.75\baselineskip}
\newcommand*{\thealgruledepth}{.25\baselineskip}
\newcount\ALG@printindent@tempcnta
\def\ALG@printindent{%
    \ifnum \theALG@nested>0
        \ifx\ALG@text\ALG@x@notext
        \else
            \unskip
            \addvspace{-1pt}
            \ALG@printindent@tempcnta=1
            \loop
                \algrule[\csname ALG@ind@\the\ALG@printindent@tempcnta\endcsname]%
                \advance \ALG@printindent@tempcnta 1
            \ifnum \ALG@printindent@tempcnta<\numexpr\theALG@nested+1\relax
            \repeat
        \fi
    \fi
    }%
\usepackage{etoolbox}
\patchcmd{\ALG@doentity}{\noindent\hskip\ALG@tlm}{\ALG@printindent}{}{\errmessage{failed to patch}}
\makeatother
\newbox\statebox
\newcommand{\myState}[1]{%
    \setbox\statebox=\vbox{#1}%
    \edef\thealgruleheight{\dimexpr \the\ht\statebox+1pt\relax}%
    \edef\thealgruledepth{\dimexpr \the\dp\statebox+1pt\relax}%
    \ifdim\thealgruleheight<.75\baselineskip
        \def\thealgruleheight{\dimexpr .75\baselineskip+1pt\relax}%
    \fi
    \ifdim\thealgruledepth<.25\baselineskip
        \def\thealgruledepth{\dimexpr .25\baselineskip+1pt\relax}%
    \fi
    \State #1%
    \def\thealgruleheight{\dimexpr .75\baselineskip+1pt\relax}%
    \def\thealgruledepth{\dimexpr .25\baselineskip+1pt\relax}%
}

\newtheorem{definition}{Definition}
\newtheorem{theorem}{Theorem}
\newtheorem{lemma}{Lemma}
\def\tsc#1{\csdef{#1}{\textsc{\lowercase{#1}}\xspace}}
\tsc{WGM}
\tsc{QE}

\usepackage[title]{appendix}

\begin{document}

\let\WriteBookmarks\relax
\def\floatpagepagefraction{1}
\def\textpagefraction{.001}



\title[mode = title]{On-the-fly Unfolding with Optimal Exploration for Linear Temporal Logic Model Checking of Concurrent Software and Systems}

\author[inst1]{Shuo Li}[orcid=0000-0002-1984-6271]
\ead{lishuo20062005@126.com}
\credit{Methodology, Formal analysis, Writing - Original Draft}

\author[inst1]{Li'ao Zheng}
\ead{zlailmm@163.com}
\credit{Software}

\author[inst2]{Ru Yang}[orcid=0000-0001-7879-681X]
\ead{yangru@shnu.edu.cn}
\credit{Writing - Review \& Editing}

\author[inst1]{Zhijun Ding}[orcid=0000-0003-2178-6201]
\ead{dingzj@tongji.edu.cn}
\cormark[1]
\cortext[1]{Corresponding author}
\credit{Conceptualization, Writing - Review \& Editing}

\affiliation[inst1]{organization={Tongji University},
    city={Shanghai},
    postcode={201804},
    country={China}
}

\affiliation[inst2]{organization={Shanghai Normal University},
    city={Shanghai},
    postcode={200234},
    country={China}
}

\begin{abstract}
\noindent\textbf{Context}: Linear temporal logic (LTL) model checking faces a significant challenge known as the state-explosion problem. The on-the-fly method is a solution that constructs and checks the state space simultaneously, avoiding generating all states in advance. But it is not effective for concurrent interleaving. Unfolding based on Petri nets is a succinct structure covering all states that can mitigate this problem caused by concurrency. Many state-of-the-art methods optimally explore a complete unfolding structure using a tree-like structure. However, it is difficult to apply such a tree-like structure directly to the traditional on-the-fly method of LTL. At the same time, constructing a complete unfolding structure in advance and then checking LTL is also wasteful. Thus, the existing optimal exploration methods are not applicable to the on-the-fly unfolding.

\noindent\textbf{Objective}: To solve these challenges, we propose an LTL model-checking method called on-the-fly unfolding with optimal exploration. This method is based on program dependence net (PDNet) proposed in the previous work. 

\noindent\textbf{Method}: Firstly, we define conflict transitions of PDNet and an exploration tree with a novel notion of delayed transitions, which differs from the existing tree-like structure. The tree improves the on-the-fly unfolding by exploring each partial-order run only once and avoiding enumerating all possible combinations. Then, we propose an on-the-fly unfolding algorithm that simultaneously constructs the exploration tree and generates the unfolding structure while checking LTL. 

\noindent\textbf{Results}: We implement a tool for concurrent programs. It also improves traditional unfolding generations and performs better than \textit{SPIN} and \textit{DiVine} on the used benchmarks.

\noindent\textbf{Conclusion}: The core contribution of this paper is that we propose an on-the-fly unfolding with an optimal exploration method for LTL to avoid the enumeration of possible concurrent combinations. This novel method outperforms traditional unfolding generation methods and the existing tools.
\end{abstract}


\maketitle

\section{Introduction}\label{Sec:Intr}
 
Model checking of concurrent software and systems is a very active topic, but it is still challenging due to the existence of state explosion. When verifying linear temporal logic (LTL), the on-the-fly method constructs and checks the state space simultaneously, rather than building all states first and then checking them. This makes the process more efficient as it only explores the parts of the state space that are necessary for verifying the property. 
However, due to the large scale of concurrent interleaving, many equivalent but unequal paths quickly aggravate the state-explosion problem. Although the generation of all states is avoided, the on-the-fly method is not effective for concurrent interleaving. 
 
Unfolding can mitigate the state-explosion problem based on the true-concurrent nature of Petri nets (PNs) \cite{McMillan1993Using}. The reason is that an unfolding is a succinct structure covering all states. It represents the partial-order semantics of PNs and can be exponentially more concise than a naive reachability graph. 
Existing unfolding-based methods \cite{Wallner1998Model,Esparza2001Implementing,Khomenko2004Parallel} check the property by unfolding a synchronization of PNs and the Büchi representation of LTL. They are inspired by the automata-theoretic approach \cite{Ding2023EnPAC}, called on-the-fly unfolding in this paper. 
However, the traditional unfolding generation \cite{Khomenko2004Parallel,Rodrguez2013An} is still inherently stateful. Once a new event is extended to the unfolding structure, all possible concurrent combinations are enumerated. This makes them need to solve an NP-complete problem \cite{Rodriguez2015Unfolding}, which seriously limits the performance of the existing unfolding generations as they grow.

As far as we know, state-of-the-art methods \cite{Rodriguez2015Unfolding,Coti2021Quasi,Schemmel2020Symbolic} optimally explore a complete unfolding structure by using a tree-like structure to verify simple safety properties of concurrent programs, such as deadlock, assertion violation, and data race.
Intuitively, it needs complete information, ensuring all partial-order runs are explored only once. A partial-order run represents a set of executions without enforcing the order on independent actions \cite{Coti2021Quasi}. And pairs of independent actions can executed in any order to produce the same state.
In existing optimal exploration methods \cite{Rodriguez2015Unfolding}, the unfolding structure is generated in advance and contains complete information on event relations.
However, the on-the-fly unfolding method generates and checks simultaneously to verify LTL, which does not contain complete information on event relations.
Thus, although it can effectively avoid enumeration combinations, it is still difficult to directly apply the existing optimal exploration using tree-like structures on a PN for LTL.
At the same time, because some parts of the state space are not necessary for verifying the property, constructing a complete unfolding structure in advance \cite{Dietsch2021Verification,Xiang2023Checking} and then checking the property is also wasteful for LTL. Thus, such optimal exploration methods are not applicable to the traditional on-the-fly method of LTL. 



In order to solve the above problems, this paper proposes an LTL model-checking method, called on-the-fly unfolding with optimal exploration, to verify the LTL properties of concurrent programs. 
The core of this method is an exploration tree which is adapted to the on-the-fly unfolding in this paper. 
As a binary tree, it consists of all partial-order runs merged with the same prefix, ensuring that each partial-order run is explored only once. When constructing an exploration tree in the existing optimal exploration methods \cite{Rodriguez2015Unfolding,Coti2021Quasi}, the child nodes of an exploration node are derived from conflict and causal relations. Since the on-the-fly unfolding of LTL does not have complete information on conflict and causal relations of events, the node in our exploration tree cannot be defined in the existing methods \cite{Rodriguez2015Unfolding,Coti2021Quasi}. 
Firstly, the concept of the disabled set from the existing methods \cite{Rodriguez2015Unfolding,Coti2021Quasi} records events that cannot occur in subsequent nodes, ensuring that a partial-order run is explored only once. 
However, if this method is applied directly to a PN, it cannot ensure the completeness of partial-order runs in the on-the-fly unfolding. This is because disabling a transition in a PN prevents the partial-order run of the later-fired transition from being explored. 
Then, generating the unfolding structure in advance is impractical in the case of LTL.

In fact, causal relations are easily captured by the order of transition occurrence, and conflict relations are defined on the transition relations in PNs. Consequently, we put forth a novel notion of the delayed set of conflict transitions in our exploration tree. 
And then, we propose an on-the-fly unfolding algorithm for constructing the exploration tree, whereby the unfolding and checking of the counterexample are carried out simultaneously. This differs from the existing exploration algorithm \cite{Rodriguez2015Unfolding,Coti2021Quasi}. 
We define the explore function to update the delayed set on the fly. 
The exploration tree guides on-the-fly unfolding to extend a new event after the addition of new nodes without enumerating all possible combinations. 
This allows for the verification of the existence of a counterexample of LTL during on-the-fly unfolding.

However, it is challenging to identify conflict relations without defining independence from concurrent programs. In our previous work, we proposed a kind of Petri net named PDNet (program dependence net) for concurrent programs \cite{2023arXiv2301Ding,Li2014Change}, which is ideally suited to describe the dependencies between actions of concurrent programs. Thus, the conflict relations based on the dependence (as a complementary of independence) of concurrent program semantics can be represented on PDNet. Thus, we propose a formal definition of conflict transition based on PDNet.

The PDNet representations for concurrent programs and the negation of LTL can be synchronized, and the LTL can be checked while unfolding a synchronization with an on-the-fly way based on our exploration tree. 
Then, the main contributions are as follows.

1. The traditional conflict relations between transitions in PNs are insufficient to cover all the conflict actions in PDNet-modeled concurrent programs. Based on dependencies from PDNet, we propose a formal definition of conflict transition to identify conflict relations. 

2. The concept of the disabled set from the existing methods \cite{Rodriguez2015Unfolding,Coti2021Quasi}, if applied to a Petri net, cannot ensure the entire partial-order runs. Based on conflict transitions, we propose an exploration tree with a new notion of delayed transition in each node, which helps optimally explore all partial-order runs and avoids enumerating all possible combinations in generating an unfolding structure.

3. We propose an on-the-fly unfolding algorithm by constructing an exploration tree and unfolding the synchronization simultaneously while checking whether a counterexample exists. Our exploration tree is constructed based on an unfolding structure, with the difficulty of updating the delayed set in new nodes using the on-the-fly unfolding. 

4. We implement an automatic tool \textit{PUPER} (PDNet Unfolding and Partial-order explorER) to verify LTL-$_\mathcal{X}$ for concurrent programs.

The next section illustrates with a motivating example. Section \ref{Sec:PDNet} introduces PDNet. Our on-the-fly unfolding with optimal exploration method is proposed in Section \ref{Sec:Unf}. Section \ref{Sec:Exp} describes the experimental evaluation. Related works are reviewed in Section \ref{Sec:Rel}. Finally, we conclude this paper.

\section{Motivating Example}\label{Sec:Mot}

\begin{figure*}[!ht]\centering
\includegraphics[width=0.95\textwidth]{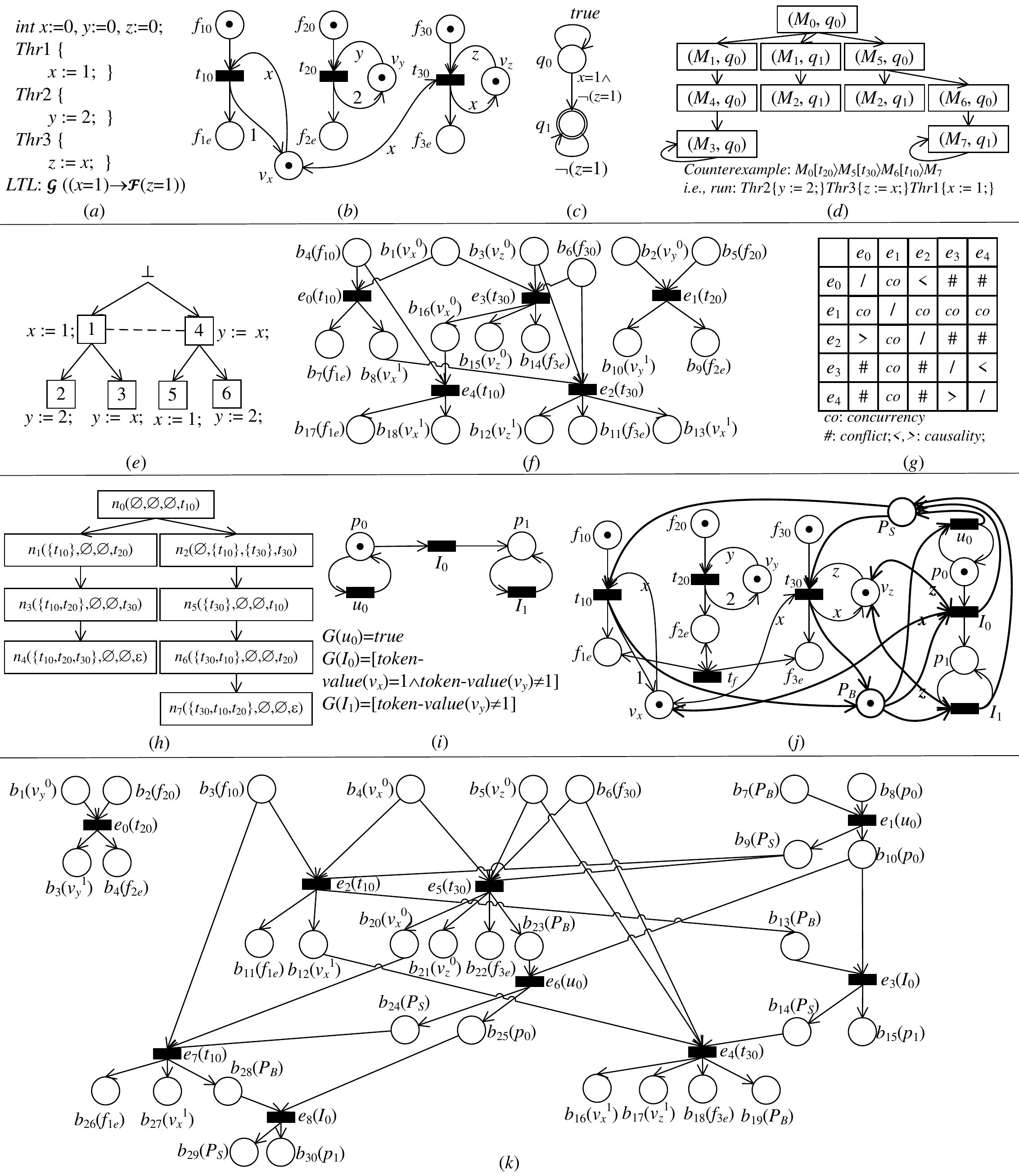}
\caption{Motivating Example. ($a$) An Example Concurrent program $\mathcal{P}$ and LTL formula ($b$) PDNet ($c$) Büchi Automaton ($d$) Model checking ($e$) The Event Structure for the example program ($f$) PDNet Unfolding ($g$) Event Relations ($h$) Exploration Tree ($i$) The Büchi PDNet ($j$) PDNet Synchronization ($k$) Unfolding of PDNet Synchronization.}\label{Fg:Motivation}
\end{figure*}

Consider the concurrent program $\mathcal{P}$ with three POSIX threads shown in Figure \ref{Fg:Motivation}($a$). There is a formula $\psi$ expressed by $\mathcal{G}(x$$=$$1)$$\rightarrow$$(\mathcal{F}(z$$=$$1))$ to specify an LTL property over two variables. $\mathcal{G}$ and $\mathcal{F}$ are two operators of LTL.
We model $\mathcal{P}$ to a PDNet $N_\mathcal{P}$, shown in Figure \ref{Fg:Motivation}($b$).
Then, the negation of this formula is translated to a Büchi automaton \cite{He2021More} $\mathcal{A}_{\neg\psi}$ in Figure \ref{Fg:Motivation}($c$). According to the traditional automata-theoretic approach, a product automaton is synchronized with an on-the-fly way to check the existence of a counterexample. In this example, we find a counterexample (an acceptable path) as a run $y:=2$, $z:=x$, $x:=1$ corresponding to $t_{20},t_{30},t_{10}$ in Figure \ref{Fg:Motivation}($d$). 

It can be observed that a partial-order run of $\mathcal{P}$ represents a run without enforcing an order of execution on independence. Based on state-of-the-art methods \cite{Rodriguez2015Unfolding,Coti2021Quasi,Schemmel2020Symbolic}, there are two partial-order runs based on unfolding semantics. These methods provide an optimal exploration of the unfolding structure, as the event structure in Figure \ref{Fg:Motivation}($e$). The node of the tree-like structure has different event sets to cover all possible partial-order runs. 
However, constructing a complete unfolding structure in advance and then checking the property is wasteful for LTL. Because some parts of the unfolding structure are unnecessary for verifying the property.



Unfolding is a succinct structure merging all partial-order runs.
Consider the unfolding of the PDNet $N_\mathcal{P}$ shown in Figure \ref{Fg:Motivation}($f$). 
There are $2$ partial-order runs via $5$ events and $18$ conditions.
The relations (causality, conflict, concurrency) between events are listed in Figure \ref{Fg:Motivation}($g$). Obviously, this unfolding preserves all concurrency from PDNet, such as $e_0$ $co$ $e_1$, meaning $t_{10}$ and $t_{20}$ are concurrent.
Thus, $\{e_0,e_1,e_2\}$ and $\{e_1,e_3,e_4\}$ as two maximal configurations capture two partial-order runs of $\mathcal{P}$ by the exploration tree in Figure \ref{Fg:Motivation}($h$).

In this paper, the Büchi automaton $\mathcal{A}_{\neg\psi}$ should be translated to a Petri net \cite{Khomenko2004Parallel}, such as the Büchi PDNet $N_{\neg\psi}$ in Figure \ref{Fg:Motivation}($i$), synchronized with $N_\mathcal{P}$.
Here, $p_1$ is an acceptable place corresponding to the acceptable state $q_1$, and the guard functions of transitions are from the arc labels in $\mathcal{A}_{\neg\psi}$.
Such synchronization as a folded net for LTL does not contain complete information on relations of events. 
Thus, the existing methods \cite{Rodriguez2015Unfolding} with tree-like structures are not applicable to the traditional on-the-fly methods of LTL.
At the same time, the traditional generation encounters an NP-complete problem when adding events to the unfolding \cite{Rodriguez2015Unfolding}. 
For example, when adding the new event $e_4$, the events that are in conflict and concurrent with $e_4$ need to be calculated, and all possible combinations $\{e_4\}$ and $\{e_1,e_4\}$ should be enumerated for possible extensions. This implies that traditional unfolding generations are also inherently stateful.
As the number of events increases, the computational complexity also rapidly increases.

Then, we define an exploration tree (as a binary tree) with a delayed transition set based on PDNet. The tree guides $N_\mathcal{P}$ to capture two partial-order runs ($n_4$ and $n_7$ as terminal nodes). Here, when $n_0$ is explored, we get the next enabled transition $t_{20}$ under $\emptyset\cup\{t_{10}\}$. And we delay $t_{10}$ and prioritize $t_{30}$ under $t_{10}\boxplus t_{30}$ ($t_{10}$ and $t_{30}$ are in conflict). Thus, $n_1$ is the left branch of $n_0$, while $n_2$ is the right one. 
When $t_{30}$ occurs, $t_{10}$ is removed from the delayed set in $n_5$. Thus, it can explore the partial-order run with the later-fired $t_{10}$ of $\mathcal{P}$.
If using the concept of the disabled set from the existing methods, the disability of $t_{10}$ cannot explore this partial-order run.
As checking $\psi$, the Büchi PDNet product $\mathcal{N}_{\neg\psi}$ is synchronized by $N_\mathcal{P}$ and $N_{\neg\psi}$ shown in Figure \ref{Fg:Motivation}($j$) (the bolded arcs for synchronization).
Then, we generate the unfolding $\mathcal{N}_{\neg\psi}$ as the exploration tree is constructed, shown in Figure \ref{Fg:Motivation}($k$). 
For instance, when $e_8$ is extended, it yields a counterexample by $e_0,e_1,e_5,e_6,e_7,e_8$ (corresponding to $t_{20},t_{30},t_{10}$) for this LTL property. That is, $\mathcal{P}\not\models\psi$.
In summary, this motivating example illustrates on-the-fly unfolding with optimal exploration presented in this paper.

\section{Program Dependence Net (PDNet)}\label{Sec:PDNet}

\subsection{Model Checking on PDNet}\label{Sub:MC}

PDNet, as a kind of Petri net \cite{BOUJARWAH1996Modelling,Yao1997Mapping}, is ideally suited to model concurrent programs and describe the dependencies between actions. The syntax of concurrent programs is given in Ding et al. \cite{2023arXiv2301Ding}.
In the following, $\mathbb{B}$ is the set of Boolean predicates with standard logic operations, $\mathbb{E}$ is a set of expressions, $Type[e]$ is the type of an expression $e$$\in$$\mathbb{E}$, i.e., the type of the values obtained when evaluating $e$, $Var(e)$ is the set of all variables in an expression $e$, $\mathbb{E}_V$ for a variable set $V$ is the set of expressions $e$$\in$$\mathbb{E}$ such that $Var(e)$$\subseteq$$V$, $Type[v]$ is the type of a variable $v$$\in$$V$, $\mathbb{O}$ is the set of constants, and $Type[o]$ is the type of constant $o$$\in$$\mathbb{O}$.

\begin{definition}[PDNet]\label{Def:PDNet}
PDNet is defined as a $9$-tuple $N$ $::=$ $(\Sigma$, $V$, $P$, $T$, $F$, $C$, $G$, $E$, $I)$, where:
	
(1) $\Sigma$ is a finite non-empty set of types called color sets.
	
(2) $V$ is a finite set of typed variables. $\forall v$$\in$$V\colon Type[v]$$\in$$\Sigma$.
	
(3) $P=P_c\cup P_v\cup P_f$ is a finite set of places.  $P_c$ is a subset of control places, $P_v$ is a subset of variable places, and $P_f$ is a subset of execution places.
	
(4) $T$ is a finite set of transitions and $T\cap P=\emptyset$.
	
(5) $F\subseteq(P\times T)\cup(T\times P)$ is a finite set of directed arcs. $F=F_c\cup F_{rw}\cup F_f$. Concretely, $F_c\subseteq(P_c\times T)\cup(T\times P_c)$ is a subset of control arcs, $F_{rw}\subseteq(P_v\times T)\cup(T\times P_v)$ is a subset of read-write arcs, and $F_f\subseteq(P_f\times T)\cup(T\times P_f)$ is a subset of execution arcs.
	
(6) $C\colon P$$\rightarrow$$\Sigma$ is a color set function that assigns a color set $C(p)$ belonging to the set of types $\Sigma$ to each place $p$.
	
(7) $G\colon T$$\rightarrow$$\mathbb{E}_V$ is a guard function that assigns an expression $G(t)$ to each transition $t$. $\forall t$$\in$$T\colon Type[G(t)]\in BOOL)\wedge(Type[Var(G(t))]\subseteq\Sigma$.
	
(8) $E\colon F$$\rightarrow$$\mathbb{E}_V$ is a function that assigns an arc expression $E(f)$ to each arc $f$. $\forall f$$\in$$F$$:$ $(Type[E(f)]=C(p(f))_{MS})\wedge(Type[Var(E(f))]\subseteq\Sigma)$, where $p(f)$ is the place connected to arc $f$.
	
(9) $I\colon P$$\rightarrow$$\mathbb{E}_\emptyset$ is an initialization function that assigns an initialization expression $I(p)$ to each place $p$. $\forall p$$\in$ $P\colon Type[I(p)]$$=$$C(p)_{MS})\wedge(Var(I(p))$$=$$\emptyset$.
\end{definition}

In the context of this paper, let $N$ be a PDNet. 
More details of PDNet and LTL formula of PDNet are in Ding et al. \cite{2023arXiv2301Ding}. 
According to the automata-theoretic approach, the LTL formula needs to be converted into an automaton and synchronized with PDNet's state space to check for the existence of a counterexample. This on-the-fly exploration constructs and checks the state space simultaneously, avoiding generating all states in advance \cite{Ding2023EnPAC}.
However, the complete interleaving of some concurrent actions brings about the path-explosion problem. Therefore, based on this idea, we adopt the conversion and synchronization to PDNet, and the PDNet synchronization (called Büchi PDNet product) is defined in Appendix \ref{App:Syn}.

As an illustration, Figure \ref{Fg:Motivation}($j$) demonstrates how to construct the PDNet nsynchronization $\mathcal{N}_{\neg\psi}$ of $N_{\neg\psi}$ and $N$, where bolded arcs are added by synchronization algorithm in Appendix \ref{App:Syn}.
For instance, $(v_x, I_0)$ and $(I_0,v_x)$ represent the observation of Büchi transition $I_0$. $(P_B,I_0)$ and $(I_0,P_S)$ represent the turn from $N_{\neg\psi}$ to $N$.


Unfortunately, not all runs that violate $\psi$ are detected by illegal infinite-trace. 
This is because $N$ is synchronized only with the visible transitions, and the invisible transitions are still concurrent.
Under the consideration of invisible transitions that cannot change the marking of observable place, the PDNet synchronization may reach a marking enabling an infinite occurrence sequence of invisible transitions without scheduling $N_{\neg\psi}$.
Let $\pi$ be a run of $\mathcal{N}_{\neg\psi}$. We can split $\pi$ into two subsequences: a finite prefix $\pi_1$ ending with the last occurrence of a visible transition and an infinite suffix $\pi_2$ containing only invisible transitions. 
$M_{n}$ is the marking reached by the occurrence of $\pi_1$. And $M_{n+k}$ is the marking reached by the occurrence of $t_{k}$ ($k>0$) in $\pi_2$.
We project the marking $M$ of $\mathcal{N}_{\neg\psi}$ onto $P_{Obs}\cup P_\mathcal{A}$, denoted by $\downarrow M$.
Since $\pi_2$ does not contain visible transitions, $\downarrow M_{n+k}$ ($k>0$) is the same as $\downarrow M_{n}$.
Let $q_n$ be the Büchi place in $\downarrow M_{n}$.
If $\mathcal{A}_{\neg\psi}$ with $q_n$ as the initial state can accept such infinite marking projection sequence $(\downarrow M_{n})^\omega$, $\pi$ is an acceptable run but not an illegal infinite-trace. 
$M_{n}$ is called a monitor marking in this run, which will be used in the following analysis.
We have to check the existence of such an acceptable run, called illegal livelock.
Thus, whether an LTL property is satisfied by a concurrent program is decided by the following theorem.

\begin{theorem}\label{The:Emp}
	Let $\mathcal{P}$ be a concurrent program, $N$ be the PDNet constructed for $\mathcal{P}$, $\psi$ be an LTL-$_\mathcal{X}$ formula, $N_{\neg\psi}$ be the Büchi PDNet of $\psi$, and $\mathcal{N}_{\neg\psi}$ be the Büchi PDNet product of $N$ and $N_{\neg\psi}$.
	$\mathcal{P}\models\psi$ iff $\mathcal{N}_{\neg\psi}$ has no illegal infinite-trace and no illegal livelock.
\end{theorem}

The intuition behind this theorem is that if $\mathcal{N}_{\neg\psi}$ has an illegal trace, there exists a run accepted by $\mathcal{A}_{\neg\psi}$. 
It implies $\mathcal{P}$ has a run violates $\psi$.
Conversely, such a run does not exist, implying $\mathcal{P}\models\psi$.
The problem is translated to the existence of illegal traces, checked during unfolding $\mathcal{N}_{\neg\psi}$ \cite{Khomenko2004Parallel}.

\subsection{PDNet Unfolding}\label{Sub:Punf}

Unfolding \cite{Esparza2002An} based on partial-order semantics is often exponentially smaller than the reachability graph. In this paper, we define unfolding for PDNet. 

Let $x$ and $y$ be two nodes (places or transitions) of a PDNet. 
If there exists a path $[x,y]$ with at least one arc, $x$ and $y$ are in causal relation, denoted by $x<y$.
If there are distinct transitions $t,t'\in T$ such that $t\neq t'$, $^\bullet t\cap^{\bullet}$$t'\neq\emptyset$, $[t,x]$ and $[t',y]$ are two paths in $N$, $x$ and $y$ are in conflict relation, denoted by $x\#y$. A node $y$ is in self-conflict if $y\#y$.
If neither $x>y$ nor $y<x$ nor $x\#y$, $x$ and $y$ are in concurrent relation, denoted by $x$ $co$ $y$.

Let $O::=(\mathcal{Q}, \mathcal{E}, \mathcal{W})$ be an occurrence net, where $\mathcal{Q}$ is a finite set of conditions (places) such that	$\forall q\in \mathcal{Q}$, $\mid^{\bullet}$$q\mid\leq 1$,
$\mathcal{E}$ is a finite set of events (transitions) such that $\forall e\in\mathcal{E}, ^{\bullet}$$e \neq\emptyset$, and
$\mathcal{W}\subseteq (\mathcal{Q}\times\mathcal{E})\cup(\mathcal{E}\times\mathcal{Q})$ is a finite set of arcs.
There are finitely many $y'$ such that $y'<y$.
Let a partial order $\prec$ be the transitive relation of $\mathcal{W}$.
$Min(O)$ denotes the set of conditions with an empty preset, which are minimal w.r.t. $\prec$.

\begin{definition}[Branching Process of PDNet]\label{Def:BP}
Let $(O,\lambda)$ be a branching process of $N$. $O=(\mathcal{Q}, \mathcal{E}, \mathcal{W})$ and $\lambda$ is a homomorphism from $O$ to $N$. $\lambda\colon \mathcal{Q}\cup \mathcal{E}\rightarrow P\cup T$ such that 	
 
(1) $\lambda(\mathcal{Q})\subseteq P$ and $\lambda(\mathcal{E})\subseteq T$, 	

 (2) $\forall e\in \mathcal{E}$, $\lambda(^{\bullet}$$e)= ^{\bullet}$$\lambda(e)$ and $\lambda( e^{\bullet})=\lambda(e)^{\bullet}$, 	

 (3) $\lambda(Min(O))=M_{0}$, and 	

 (4) $\forall e,e'\in\mathcal{E}$, if $^{\bullet}$$e=^{\bullet}$$e'$ and $\lambda(e)=\lambda(e')$, then $e=e'$.
\end{definition}

There is no redundancy in $(O,\lambda)$, and the environments of all transitions are preserved.
If $\lambda(e)$$=$$t$, $e$ is called $t$-labelled event.
$Min(O)$ corresponds to the initial marking of $N$.
Let $\hbar::=(O,\lambda)$ be a branching process of $N$.
Branching processes of the same PDNet differ on the scales by unfolding.
We define a prefix relation $\sqsubseteq$ which formalizes that one branching process is less than another by unfolding the same PDNet.
A branching process $\hbar'=(O',\lambda')$ of $N$ is a prefix of another branching process $\hbar=(O,\lambda)$, denoted by $\hbar'\sqsubseteq\hbar$, if $O'=(\mathcal{Q}',\mathcal{E}',\mathcal{W}')$ is a subnet of $O$ containing all minimal elements in $Min(O)$, and such that
(1) $\forall e\in \mathcal{E}', (q,e)\in \mathcal{W}\vee(e,q)\in \mathcal{W}$ $\Rightarrow$ $q\in \mathcal{Q}'$, and
(2) $\forall q\in \mathcal{Q}'$, $(e,q)\in\mathcal{W}$ $\Rightarrow$ $e\in \mathcal{E}'$, and
(3) $\lambda'$ is the restriction of $\lambda$ to $\mathcal{Q}'\cup \mathcal{E}'$.
There always exists a unique maximal branching process w.r.t. $\sqsubseteq$ called the unfolding, denoted by $Unf_{N}$.
To formulate how $Unf_{N}$ describes the behavior of a PDNet, we define the notions of configuration and cut.

\begin{definition}[Configuration]\label{Def:CoCut}
Let $\hbar$ be a branching process of $N$. A configuration $\mathcal{C}$ of $\hbar$ is a finite set of events such that
	
(1) $\forall e\in\mathcal{C}$, $e'\prec e$ $\Rightarrow$ $e'\in\mathcal{C}$ ($\mathcal{C}$ is causally closed), and
	
(2) $\forall e,e'\in\mathcal{C}$, $\neg(e\#e')$ ($\mathcal{C}$ is conflict-free).
\end{definition}

Let $\mathcal{Q}$ be a finite set of conditions. $\mathcal{Q}$ is a \textit{co}-set of $\hbar$ if $\forall q,q'\in\mathcal{Q}$, $q\neq q'$ $\wedge$ $q$ $co$ $q'$. That is, any two distinct conditions are pairwise in concurrent relation.
Then, a maximal \textit{co}-set w.r.t. $\subseteq$ is a cut.
For $e\in\mathcal{E}$, $[e]::=\{e'\vert e'\preceq e\}$ is the local configuration of $e$.
Then, the cut of a configuration $\mathcal{C}$ is
$Cut(\mathcal{C})$$=$$Min(O)\cup\left(\bigcup\limits_{e\in\mathcal{C}} e^{\bullet}\right) \setminus \left(\bigcup\limits_{e\in \mathcal{C}}{^{\bullet}e}\right)$.
And $Mark(\mathcal{C})$$=$$\lambda(Cut(\mathcal{C}))$ represents a marking of $N$, called the final marking w.r.t. $\mathcal{C}$.
Loosely speaking, $Mark(\mathcal{C})$ is the reachable marking from $M_0$ by occurring the events in $\mathcal{C}$ exactly once.
$\mathcal{C}$ has one or more event sequences (called linearizations) that lead to the same marking. By applying $\lambda$ to a linearization, we get an occurrence sequence of $N$.

For a configuration $\mathcal{C}$ and an event set $\mathcal{E}'$, the fact that $\mathcal{C}\cup \mathcal{E}'$ is a configuration and $\mathcal{C}\cap \mathcal{E}'=\emptyset$ is denoted by $\mathcal{C}\oplus \mathcal{E}'$.
Therefore, such a set $\mathcal{E}'$ is a suffix of $\mathcal{C}$, and $\mathcal{C}\oplus \mathcal{E}'$ is an extension of $\mathcal{C}$.
For a configuration $\mathcal{C}'$, if $\mathcal{C}\subset\mathcal{C}'$, there exists a non-empty suffix $ \mathcal{E}$ of $\mathcal{C}$ such that $\mathcal{C}\oplus\mathcal{E}=\mathcal{C}'$.
Then, we define how to extend the current branching process.

\begin{definition}[Possible Extension of PDNet Unfolding]\label{Def:PE}
Let $\hbar=(O,\lambda)$ be a prefix of $Unf_N$. 
A possible extension of $\hbar$ is a pair $(t,\chi)$, where 

(1) $\chi$ is a \textit{co}-set in $\hbar$, such that $\lambda(\chi)=^{\bullet}$$t$, and 

(2) $\hbar$ contains no $t$-labelled event with a preset $\chi$.
\end{definition}

The pair $(t,\chi)$ is an event captured by unfolding algorithms \cite{Esparza2002An} to extend the prefix $\hbar$ of $Unf_N$.
A new prefix $\hbar'$ is extended by adding all possible extensions to $\hbar$, and $\hbar\sqsubseteq\hbar'$. 
However, traditional unfolding generation methods are inherently stateful and must enumerate all possible combinations for $\hbar$, which seriously limits the performance as $\hbar$ grows. Although it is possible to prune the entire combinations by defining conditions that must not be enumerated, there are significant limitations to these conditions \cite{Rodrguez2013An}.
It is greatly limited for the practical application of concurrent programs, especially when the concurrent program exhibits high concurrency. 

\subsection{Cutoff of PDNet}\label{Sub:Cutoff}

Although $Unf_N$ is always infinite, there is always at least one finite complete prefix of $Unf_N$.
The idea \cite{Esparza2002An} is that if the prefix contains two events, then the continuations of their local configurations are isomorphic. Thus, only one of them needs to be explored further. 
Let $\mathcal{C}_1$ and $\mathcal{C}_2$ be two finite configurations leading to the same marking, i.e., $Mark(\mathcal{C}_1)=Mark(\mathcal{C}_2)$.
There exists an isomorphism $\theta$ mapping from the finite extensions $\mathcal{C}_1\oplus \mathcal{E}$ to the finite extensions of $\mathcal{C}_2\oplus \theta(\mathcal{E})$.

\begin{definition}[Adequate Order]\label{Def:AO}
A partial order $\prec$ on the finite configuration of $Unf_N$ is an adequate order if
	
(1) $\prec$ is well-founded,
	
(2) $\mathcal{C}_1\subset\mathcal{C}_2$ implies $\mathcal{C}_1\prec\mathcal{C}_2$, and
	
(3) $\prec$ is preserved by finite extensions. 
If $\mathcal{C}_1\prec\mathcal{C}_2$ and $Mark(\mathcal{C}_1)=Mark(\mathcal{C}_2)$, then an isomorphism $\theta$ satisfies $\mathcal{C}_1\oplus \mathcal{E}\prec\mathcal{C}_2\oplus \theta(\mathcal{E})$ for all finite extensions.
\end{definition}

Based on $Min(O)$, $Unf_N$ is generated by adding new events following the adequate order.
Then, we formally define the cut-off event of PDNet unfolding w.r.t. $\prec$.

\begin{definition}[Cut-off of PDNet Unfolding]\label{Def:Cut}
Let $\prec$ be an adequate order on the finite configuration of $Unf_N$, $\hbar$ a prefix of $Unf_N$, and $e$ an event in $\hbar$. 
The event $e$ is a cut-off event, if there exists an event $e'$ (called companion) in $\hbar$ such that $Mark([e])=Mark([e'])$, and either
 (1) $e'<e$, or
 (2) $\neg(e'<e)$ $\wedge$ $[e']\prec[e]$.

 $e$ is a successful cut-off event, if

(\uppercase\expandafter{\romannumeral1}) $[e]\setminus[e']\cap T_I \neq\emptyset$, or

(\uppercase\expandafter{\romannumeral2}) $\exists e_f\in [e]\colon\lambda(e_f)=t_f$, such that $Mark([e'])$ is a monitor marking.
\end{definition}


Once a cut-off event is extended, there are two kinds of successful cut-off events to check illegal traces.


\begin{theorem}\label{The:IllTra}
    Let $N$ be a PDNet, $\psi$ be an LTL-$_\mathcal{X}$ formula, $\mathcal{N}_{\neg\psi}$ be the Büchi PDNet product, and $\hbar$ be a prefix of $Unf_\mathcal{N}$.
\begin{enumerate}
    \item[$\bullet$] $\mathcal{N}_{\neg\psi}$ has an illegal infinite-trace iff $\hbar$ has a successful cutoff-\uppercase\expandafter{\romannumeral1}.
    \item[$\bullet$] $\mathcal{N}_{\neg\psi}$ has an illegal livelock iff $\hbar$ has a successful cutoff-\uppercase\expandafter{\romannumeral2}.
\end{enumerate}
\end{theorem}

The correctness is proven in Appendix \ref{App:Proof}. 
The complete prefix construction of $Unf_\mathcal{N}$ terminates when no new event having no cutoff events among its predecessors can be extended. 
The prefix of $Unf_\mathcal{N}$ cut off by such events satisfying completeness and finiteness \cite{Esparza2002An} ensures that the whole checking process can be terminated, even if no illegal trace exists.
$\hbar$ contains at most $S^2$ non-cutoff events, where $S$ is the number of reachable markings of $\mathcal{N}_{\neg\psi}$. 

\section{On-the-fly Unfolding with Optimal Exploration}\label{Sec:Unf}

\subsection{Conflict Transition of PDNet}\label{Sub:ConT}

Advances in formal verification methods \cite{Rodriguez2015Unfolding,Coti2021Quasi} have enabled the optimal exploration via a tree-like structure. This tree-like structure is constructed from an unfolding structure (called event structure) with complete information on the relations of events. The child nodes in this tree-like structure are derived from causal and conflict relations for updating enabled events and conflict ones. In PNs, causal relations are easily captured by the order of transition occurrence, and conflict relations are defined on the transition structure.
However, traditional conflict relations between transitions in PNs cannot cover all the conflict actions from concurrent program semantics.
Thus, it is insufficient to identify conflict transitions without defining independence in concurrent programs. We apply such independence from concurrent program semantics to guide the formal definition of conflict transitions and calculate these conflicts during modeling. 

So far, partial-order techniques use independence \cite{Rodriguez2015Unfolding} to avoid exploring equivalent execution belonging to the same partial-order run. 
The independent relation between actions is an under-approximation of commutativity \cite{Rodriguez2015Unfolding}. Intuitively, an independent relation is valid if any pair of independent actions can be executed in any order and still produce the equivalent executions w.r.t. the LTL-$_\mathcal{X}$ formula. 
A common idea for identifying independent relations is determining the dependence between actions. This is because two actions are independent if they are not dependent. 
PDNet describes the dependence (as a complement of independence) between actions of concurrent programs. 
Thus, we propose a formal definition of conflict transitions to identify conflict relations.

\begin{table}[h]\centering
\caption{PDNet Transition.}\label{Tab:PDNetTransition}
\begin{tabular}{ccc}
\toprule[1pt]
Action & Transition  & Operation \cite{2023arXiv2301Ding}   \\ \midrule
$ass$  & $assign$ transition                  & $\nu:=w$ \\ \midrule
$jum$  & $jump$ transition                    & \multirow{2}{*}{$jump$}   \\
$ret$  & $exit$ transition                    &    \\ \midrule
$tcd$  & \multirow{2}{*}{$branch$ transition} & if($w$)then($\tau_1*$)else($\tau_2*$) \\
$fcd$  &                                      & while($w$)do($\tau*$)   \\ \midrule
$call$ & $call$ transition                    & \multirow{2}{*}{$calls$}   \\
$rets$ & $return$ transition                  &  \\ \midrule
$acq$  & $lock$ transition                    & $\langle lock, \ell\rangle$  \\ \midrule
$rel$  & $unlock$ transition                  & $\langle unlock, \ell\rangle$ \\ \midrule
$sig$  & $signal$ transition                  & $\langle signal, \gamma \rangle$   \\ \midrule
$wa_1$ & \multirow{3}{*}{$wait$ Transition}   & \multirow{3}{*}{$\langle wait, \gamma, \ell \rangle$}   \\
$wa_2$ &  \\
$wa_3$ &  & \\
\bottomrule[1pt]
\end{tabular}
\end{table}
In this paper, we use $ass$, $tcd$, $fcd$, $acq$, $rel$, $wa_1$, $wa_2$, $wa_3$ and $sig$ as the abbreviations of the actions. Each of these actions corresponds to a specific PDNet transition as shown in Table \ref{Tab:PDNetTransition}.
We have informally explained how actions with different effects interleave with each other, i.e., the dependent relation between actions in Appendix \ref{App:Succ}. We consider two actions to be dependent on each other if both belong to the same thread or are interfered with by a shared variable/mutex/condition variable. 
As shown in Figure \ref{Fg:Motivation}, $x$$:=$$1$ and $z$$:=$$x$ are dependent because there is a data race, while $y$$:=$$2$ and $z$$:=$$x$ ($y$$:=$$2$ and $x$$:=$$1$) are independent. $t_{10}$ and $t_{30}$ are conflict transitions (denoted by $t_{10}\boxplus t_{30}$), meaning the dependence between $x$$:=$$1$ and $z$$:=$$x$.

Firstly, we introduce two sets for a transition to reference or define a variable of PDNet.

\begin{definition}[Reference Set and Definition Set of PDNet]\label{Def:DefRef}
Let $t$ be a transition of $N$.
The reference set of $t$ is defined by $Ref(t) ::= \{p\mid\forall p\in ^\bullet$$t\cap P_v: E(p,t) = E(t,p)\}$. 
The definition set of $t$ is defined by $Def(t) ::=
	\{p\mid\forall p \in ^\bullet$$t\cap P_v: E(p,t)\neq E(t,p)\}$.
\end{definition}

Here, $P_v$ is the variable place set of $N$ from Definition \ref{Def:PDNet}. There are two arcs (e.g., $(p,t)$ and $(t,p)$) between the variable place $p$ and the transition $t$, as long as $p\in^\bullet$$t$ based on the modeling rules \cite{2023arXiv2301Ding}. Then, we define conflict transitions.
 
\begin{definition}[Conflict Transitions of PDNet]\label{Def:Conf}
Let $t_1$ and $t_2$ be two transitions of a PDNet $N$. 
$\exists f_1$$\in$$^\bullet$$t_1$$\cap$$P_f$, $\exists f_2$$\in$$ ^\bullet$$t_2$$\cap$$P_f$ such that $M(f_1)$$\neq$$M(f_2)$ ($t_1$ and $t_2$ belong to the different threads).
$t_1$ and $t_2$ are conflict, if

(1) $t_1$ is $assign$ transition, $t_2$ is $assign$ (or $branch$) transitions, and
	$\exists p\in P_v$ such that $p\in Def(t_1)$ $\wedge$ ($p\in Def(t_2)$ $\vee$ $p\in Ref(t_2)$).
	
(2) $t_1$ is $lock$ transition or $wait$ transition for $wa_3$, 	
	(\romannumeral1) $t_2$ is $lock$ transition, and $\exists p_{\ell}\in P$ such that $p_{\ell}\in ^\bullet$$t_1$ $\wedge$ $p_{\ell}\in ^\bullet$$t_2$, or 	
	(\romannumeral2) $t_2$ is $unlock$ transition, and $\exists p_{\ell}\in P$ such that $p_{\ell}\in ^\bullet$$t_1$ $\wedge$ $p_{\ell}\in t_2$$^\bullet$, or	
	(\romannumeral3) $t_2$ is $wait$ transition, and $\exists p_{\ell}\in P$ such that $p\in ^\bullet$$t_1$ $\wedge$ ($p_{\ell}\in ^\bullet$$t_2$  $\vee$ $p_{\ell}\in t_2$$^\bullet$).
	
(3) $t_1$ is $unlock$ transition or the $wait$ transition for $wa_1$, 	
	(\romannumeral1) $t_2$ is $lock$ transition, and $\exists p_{\ell}\in P$ such that $p_{\ell}\in t_1$$^\bullet$ $\wedge$ $p_{\ell}\in ^\bullet$$t_2$, or	
	(\romannumeral2) $t_2$ is $unlock$ transition, and $\exists p_{\ell}\in P$ such that $p_{\ell}\in t_1$$^\bullet$ $\wedge$ $p_{\ell}\in t_2$$^\bullet$, or 	
	(\romannumeral3) $t_2$ is $wait$ transition, and $\exists p_{\ell}\in P$ such that $p_{\ell}\in t_1$$^\bullet$ $\wedge$ ($p_{\ell}\in ^\bullet$$t_2$ $\vee$ $p\in t_2$$^\bullet$).
	
(4) $t_1$ is the $wait$ transition for $wa_2$, $t_2$ is $signal$ transition, and
	$\exists p_{\gamma}\in P$ such that $p_{\gamma}\in ^\bullet$$t_1$ $\wedge$ $p_{\gamma}\in t_2$$^\bullet$.
	
(5) $t_1$ is $signal$ transition, 	
	(\romannumeral1) $t_2$ is the $wait$ transition for $wa_2$, and $\exists p_{\gamma}\in P$ such that $p_{\gamma}\in t_1$$^\bullet$ $\wedge$ $p_{\gamma}\in ^\bullet$$t_2$, or 	
	(\romannumeral2) $t_2$ is $signal$ transition, and $\exists p_{\gamma}\in P$ such that $p_{\gamma}\in t_1$$^\bullet$ $\wedge$ $p_{\gamma}\in t_2$$^\bullet$.
\end{definition}

\begin{figure}[t]\centering
\includegraphics[width=0.4\textwidth]{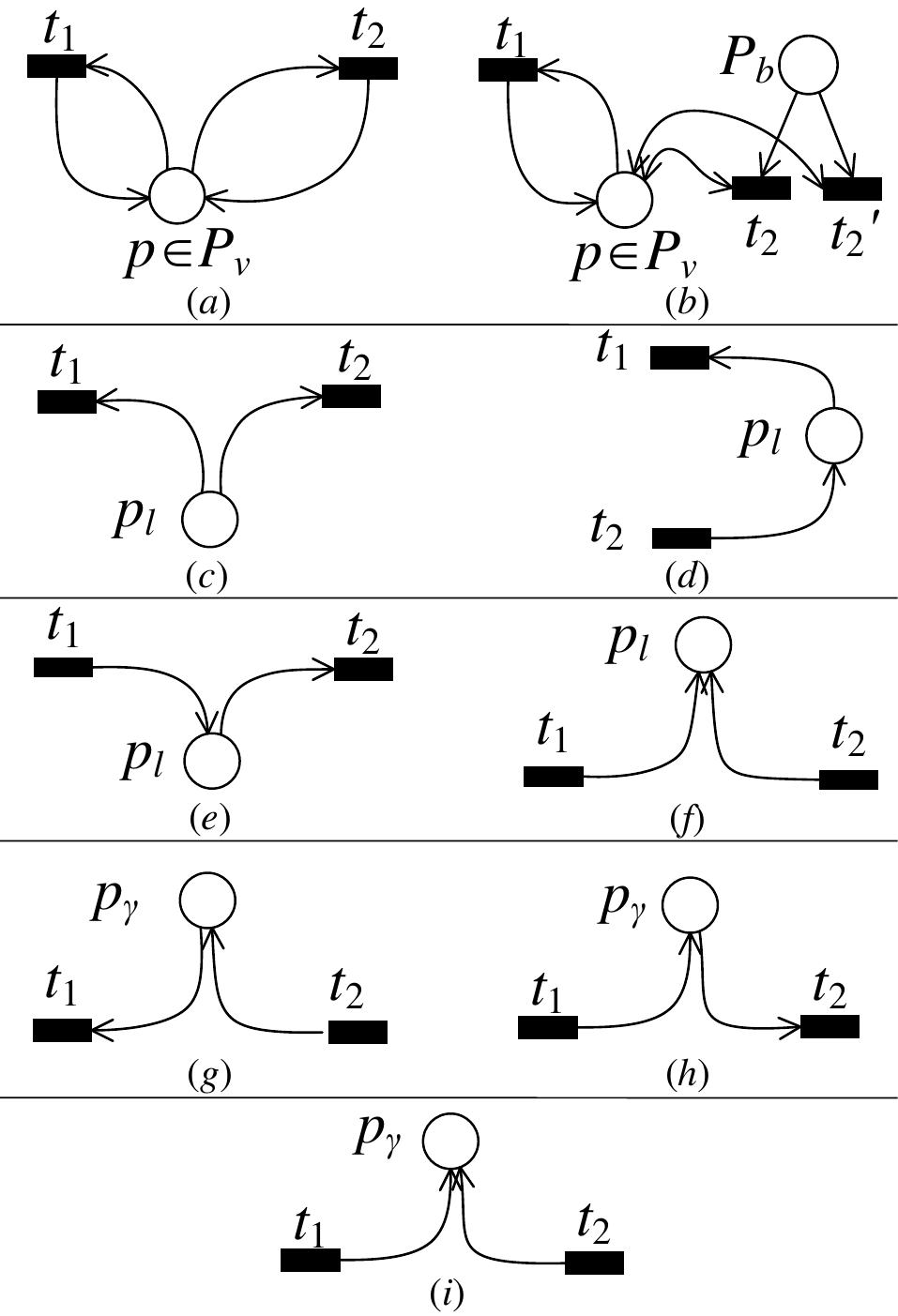}
\caption{Conflict Transitions of PDNet. 
($a$) $t_1$ is $assign$ transition and $t_2$ is $assign$ transition. ($b$) $t_1$ is $assign$ transition and $t_2$ is $branch$ transition. 
($c$) $t_1$ is $lock$ transition or $wait$ transition for $wa_3$ and $t_2$ is $lock$ transition or $wait$ transition of $wa_3$. ($d$) $t_1$ is $lock$ transition or $wait$ transition for $wa_3$ and $t_2$ is $unlock$ transition or $wait$ transition for $wa_1$. 
($e$) $t_1$ is $unlock$ transition or the $wait$ transition for $wa_1$ and $t_2$ is $lock$ transition or $wait$ transition of $wa_3$. ($f$) $t_1$ is $unlock$ transition or the $wait$ transition for $wa_1$ and $t_2$ is $unlock$ transition $wait$ transition of $wa_1$. 
($g$) $t_1$ is the $wait$ transition for $wa_2$ and $t_2$ is $signal$ transition. 
($h$) $t_1$ is $signal$ transition and $t_2$ is the $wait$ transition for $wa_2$. ($i$) $t_1$ is the $signal$ transition and $t_2$ is $signal$ transition.} \label{Fg:Conflict}
\end{figure}

Intuitively, conflicts are used to merge the prefix-sharing partial-order runs. An illustration for each case in Definition \ref{Def:Conf} is shown in Figure \ref{Fg:Conflict}. 
Then, we define a valid independence based on the conflict transitions of PDNet.

\begin{definition}[Valid Independence]\label{Def:VI}
Let $N_\mathcal{P}$ be the PDNet of a concurrent program $\mathcal{P}$.
The dependent relation $\boxplus_\mathcal{P}\subseteq T\times T$ is a reflexive, symmetric relation. $t\boxplus_\mathcal{P} t'$ hold if $t$ matches the transition described by $t_1$ and $t'$ matches the transition described by $t_2$ in Definition \ref{Def:Conf}.
Then, the valid independent relation is the complementary relation of $\boxplus_\mathcal{P}$, denoted by $\boxtimes_\mathcal{P} ::= (T\times T)\setminus\boxplus_\mathcal{P}$.
\end{definition}

As the example in Figure \ref{Fg:Motivation}($b$), $t_{10}\boxplus_\mathcal{P}t_{30}$. The reason is that $v_x\in Def(t_{10})\vee v_x\in Ref(t_{30})$,  satisfying Definition \ref{Def:Conf}(1). 
Therefore, we can obtain all necessary conflict transitions for our exploration tree of $\mathcal{N}_{\neg\psi}$. 

\subsection{Exploration Tree for PDNet}\label{Sub:Ter}

The existing methods \cite{Rodriguez2015Unfolding,Coti2021Quasi} propose a tree-like structure for optimal exploration. In this structure, the disabled set in a node records events that cannot occur in subsequent nodes, which can ensure that all partial-order runs are explored only once. 
However, if such a set is applied directly to $\mathcal{N}_{\neg\psi}$, it cannot ensure the completeness of partial-order runs. 
This is because disabling a transition prevents the partial-order run of the later-fired transition from being explored. 
Thus, we propose an exploration tree that differs from the existing methods. 

We define a new notion of delayed transition in each node of an exploration tree for PDNet.
Our exploration tree helps explore each partial-order run of PDNet only once. 
Intuitively, the node represents both the exploration of a configuration of PDNet unfolding and a choice to occur an enabled transition for the left child node as well as a conflict transition for the right child node.
Once a new event is extended to the unfolding structure, all possible concurrent combinations are enumerated by the traditional unfolding generation \cite{Rodrguez2013An}. Our exploration tree also avoids enumerating all combinations while generating an unfolding structure in this paper.
Next, we define the exploration node in our exploration tree.

\begin{definition}[Exploration Node on PDNet]\label{Def:ENode}
Let $\psi$ be an LTL-$_\mathcal{X}$ formula, $\mathcal{N}_{\neg\psi}$ the Büchi PDNet product, $Unf_N$ the unfolding of $\mathcal{N}_{\neg\psi}$, and $\hbar$ the current prefix of $Unf_N$.
An exploration node is a tuple $n::=(\epsilon, \kappa, \eta, t)$, where 
 
(1) $\epsilon$ is a finite set of events (configuration) of $\hbar$, 

(2) $\kappa$ is a finite set of transitions (delayed set) of $N$, 

(3) $\eta$ is a finite set of transitions (guide set) of $N$, 

(4) $t$ is an enabled transition under $Mark(\epsilon)$.
\end{definition}

In fact, $\epsilon$ is a configuration of $\hbar$ satisfying Definition \ref{Def:CoCut} that stores the information about the current partial-order run. $t$ must be enabled under $Mark(\epsilon)$ and we can extend $t$-labeled event(s) in $\hbar$.
Unlike the disabled set of events \cite{Rodriguez2015Unfolding}, the delayed set $\kappa$ is identified as conflict transitions. 
The intuition is that the transitions in $\kappa$ are delayed to explore the partial-order run of such later-fired transitions.
Therefore, updating $\kappa$ becomes a difficulty in constructing an exploration tree. 
This is based on the guide set $\eta$. $\eta$ is precisely used to guide the exploration of the next partial-order run. Intuitively, if there is a delayed transition of $t$, the prioritized and causal transitions need to be added to $\eta$. 
If a transition in $\eta$ has occurred as an extended event to the unfolding structure, the occurred transition is removed from $\eta$. At the same time, if all transitions associated with a delayed transition have occurred, the transition is removed from $\kappa$.

Based on our exploration node, the key insight of our exploration tree is as follows. 
Each node $n$ contains at most two successor nodes. 
The left child node of $n$ explores all partial-order runs, including $\epsilon\cup \{e\}$ where $\lambda(e)=t$.
And the right child node of $n$ explores the partial-order runs that $t$ is delayed to occur if there exists another transition that conflicts with it based on Definition \ref{Def:Conf}. 
In fact, the exploration tree explores all possible maximal configurations. 
These maximal configurations correspond to all possible partial-order runs. Thus, the exploration tree guides precisely the on-the-fly unfolding with a minimal number of paths.

Above all, the tree-like structure from existing methods \cite{Rodriguez2015Unfolding,Coti2021Quasi} are inherently different from the Definition \ref{Def:ENode} in this paper. 
Without complete conflict and causal relations, an enabled transition is generated using the current prefix. We extend events to the unfolding structure while constructing the exploration nodes. 
Based on this insight, our algorithm simultaneously constructs our exploration tree and generates our unfolding structure while checking the property.

\subsection{On-the-fly Unfolding Algorithm}\label{Sub:POC}

It is unnecessary to construct a complete unfolding structure in advance and then check a counterexample for LTL \cite{Dietsch2021Verification,Xiang2023Checking}. 
We propose an on-the-fly unfolding method using optimal exploration. The core is an exploration tree adapted to the on-the-fly unfolding for LTL in this paper. Optimal exploration implies that all partial-order runs are explored only once.
A partial-order run of $\mathcal{P}$ represents a set of executions without enforcing an order on independent actions \cite{Coti2021Quasi}.

Based on the above insights, we propose an algorithm in which an exploration tree is constructed while unfolding $\mathcal{N}_{\neg\psi}$ to verify the LTL properties of concurrent programs.
The information to construct the exploration nodes and extend new left and right child nodes for the continuous exploration is from the current prefix of $\mathcal{N}_{\neg\psi}$, and the guidance to extend the next event forming a full partial-order run or identify the next partial-order run is from our exploration tree. The difficulty is updating the delayed set in new nodes using the incomplete prefix. 
We explore all partial-order runs of $\mathcal{P}$ in Algorithm \ref{Alg:UPOC} and check illegal traces in Algorithm \ref{Alg:Cutoff}. 

\begin{breakablealgorithm}
\caption{On-the-fly Unfolding with optimal exploration}\label{Alg:UPOC}
		\begin{algorithmic}[1]
			\Require The Büchi PDNet product $\mathcal{N}_{\neg\psi}$ of $N$ and $\psi$;
			\Ensure $false$ or $true$;
			\State $result := true$;
			\State $\hbar := Min(O)$;
			\State $n_0 :=$ \Call {NOD}{$\emptyset,\emptyset,\emptyset,\bot$};
			\State \Call {EXPLORE}{$n_0$};
			\Procedure {EXPLORE} {$n$}
			\State $n := \langle \epsilon,\kappa,\eta,t\rangle$;
			\If {$\kappa\cap Conf(t)\neq\emptyset \wedge Fired(t)$}
			\State $\kappa:=\kappa\setminus \{t'\mid \forall t'\in\kappa\colon t\#t'\}$;
			\EndIf
			\If {$EXI(\epsilon)$} \Return
			\EndIf
			
			\State $EN(\epsilon) :=$ \Call{EXTEND}{$\epsilon$};
                \State $TEN(\epsilon) := \lambda(EN(\epsilon))$;
			
			\If {$\eta\neq\emptyset$}
			\State select a $t'$ from $TEN(\epsilon)\cap\eta$;
			
			\Else
			\State select a $t'$ from $TEN(\epsilon)\setminus\kappa$;
			\EndIf
			
			\If {$e$$:=$\Call{AddEvent}{$t,\hbar$}} \State $LChild:=true$;
			\If {\Call{CUTOFF}{$e,\hbar$}} 
                    \Return
			\EndIf
			\EndIf
			
			\If {$LChild$}
			\State $n':=$ \Call {NOD}{$\epsilon\cup\{e\},\kappa,\eta,t'$};
			\State \Call{EXPLORE}{$n'$}; 
			\EndIf
			\State $ConJ :=$ \Call {ALT}{$\epsilon,\kappa\cup\{t\},\eta$}; 
			\If {$ConJ\neq\emptyset$}
			\State $n'':=$ \Call {NOD}{$\epsilon,\kappa\cup\{t\},ConJ\setminus \epsilon$};
			\State \Call{EXPLORE}{$n''$}; 
			\EndIf
			\EndProcedure
			
			\Function {ALT}{$\epsilon,\kappa,\eta$}
			\ForAll {$t\in \kappa\cap TEN(\epsilon)$}
			\If {$t'\in Conf(t) \wedge (t'\notin\eta) \wedge (t'\notin\kappa)$} 
			\State \Return  $\eta\cup\lceil t'\rceil$; 
			\Else
			\State \Return $\emptyset$;
			\EndIf
			\EndFor
			\EndFunction
			
			
			
			\Function {EXTEND} {$\epsilon$}
			\State $pe :=$ \Call{PE}{$\epsilon,\mathcal{N}_{\neg\psi}$}; /$\ast$ Function $PE$ calculate the possible extensions$\ast$/
			\State \Return $pe$;
			\EndFunction
		\end{algorithmic}
\end{breakablealgorithm}

Our exploration tree and complete prefix of $\mathcal{N}_{\neg\psi}$ are constructed using a fixed-point loop until the tree cannot be expanded.
$\hbar$ represent the current prefix for $\mathcal{N}_{\neg\psi}$, with $Min(O)$ contains minimal conditions.
The initial node $n_0$ is constructed by the root node with an empty transition symbol $\bot$ of the exploration tree (Line 3 in Algorithm \ref{Alg:UPOC}). 
Then, the procedure EXPLORE is recursively called to build our exploration tree.
In procedure EXPLORE, $\langle\epsilon,\kappa,\eta,t\rangle$ represents a node in the exploration tree (Line 6). 

Firstly, it is necessary to check whether the delayed set $\kappa$ needs to be updated.
Function $Conf(t)$ can check if there is a conflict transition $t'$ in $\kappa$ (Line 7). 
If the last transition of $\kappa$ has occurred detected by $Fired(t)$, $t'$ should be deleted from $\kappa$ (Line 8).
Obviously, if $\epsilon$ is a duplicate configuration by $EXI(\epsilon)$, $n$ will not be explored (Line 9).
$EN(\epsilon)$ stores all possible extensions under $\epsilon$ (Line 10), while $TEN(\epsilon)$ is the enabled transition set by $EN(\epsilon)$ (Line 11). 
There are two cases to select an enabled transition. 
(\romannumeral1) $\eta$ is not empty (Line 12). We select an enabled transition from $\eta$ (Line 13).
(\romannumeral2) Otherwise, we select an enabled transition not in $\kappa$ (Line 15).
At this time, a $t$-labelled event $e$ is extended to the prefix $\hbar$ by $AddEvent(t,\hbar)$ (Line 16).
This implies $n$ has a left child node (Line 17). If $e$ is a cutoff event, $\epsilon$ is cut off (Line 18).

Secondly, a left or right child node is added to generate all possible partial-order runs. 
(\romannumeral1) $n'$ is the left child node, denoted by $n\triangleright_l n'$, which explores all partial-order runs that include $\epsilon\cup\{e\}$, and the transitions in $\kappa$ are delayed to occur (Lines 19-21). 
(\romannumeral2) If $ConJ=\emptyset$, the right child node of $n$ is empty (Line 22). Otherwise, $n''$ is the right child node, denoted by $n\triangleright_r n''$, which explores those partial-order runs to occur a transition $t'$ that conflicts with $t$ by $\kappa\cup\{t\}$ (Lines 23-25). 

Importantly, function ALT returns a set of transitions that witnesses the existence of some partial-order runs that extend $\epsilon$ when delaying some transitions from $\kappa\cup\{t\}$.
It selects a transition from $\kappa$ and $TEN(\epsilon)$ (Line 27), and searches for a transition $t'$ in conflict with $t$ (Line 28).
For such a transition $t'$, some causal transitions $\lceil t'\rceil$ need to occur before $t'$ are added to $\eta$ for the right child node (Line 29). 
Otherwise, it returns $\emptyset$ (Line 31). 
Function EXTEND calculates all possible extensions under the configuration $\epsilon$ (Line 33), 
When the $t$-labeled event $e$ is extended to $\hbar$ (Line 16), function CUTOFF in Algorithm \ref{Alg:Cutoff} is called to check whether it is a cutoff event (Line 18).
Then, we explain CUTOFF to check illegal traces.

\begin{breakablealgorithm}
\caption{Function CUTOFF}\label{Alg:Cutoff}
\begin{algorithmic}[1]
    \ForAll {$e'\in\hbar\wedge Mark([e])=Mark([e'])$} 
        \If {$e_f\notin [e]$}
            \If {$e'<e$}
                \If {$[e]\setminus[e']\cap T_I\neq \emptyset$}
                    \State $result:=false$;
                    \State\Return $true$;
                \ElsIf {$\mid[e']\cap T_I\mid\geq\mid[e]\cap T_I\mid$} 
                    \State\Return $true$;
                \EndIf
            \EndIf
        \Else
            \State $M:=Mark([e]\setminus \{e\})$;
            \State Let $M=(q,P_B,O,H)$;
            \If {$\mathcal{A}_{\neg\psi}^q \models O^\omega$}
                \State $result:=false$;
            \EndIf
            \State \Return $true$;
        \EndIf
    \EndFor
\end{algorithmic}
\end{breakablealgorithm}

As mentioned above, cutoff events prevent the exploration of some redundant configurations to ensure that $\hbar$ covers all partial-order runs and checks whether there is an illegal trace.
The function CUTOFF returns $true$ if the event $e$ is a cutoff event.
Firstly, the companion $e'$ of $e$ satisfies $Mark([e])=Mark([e'])$, implying that $[e']$$\prec$$[e]$ holds.
According to Definition \ref{Def:Cut}, there are two cases of illegal traces. 

(1) If there is no $t_f$-labelled event in $[e]$, 
we check whether an illegal infinite-trace exists. 
If $\exists t\in T_I$ belongs to $[e]\setminus [e']$ (Line 4 in Algorithm \ref{Alg:Cutoff}), $e$ is cutoff-\uppercase\expandafter{\romannumeral1}, and $result=false$ (Line 5). This implies there exists an illegal infinite-trace by Theorem \ref{The:IllTra}.
If $e'<e$ $\wedge$ $\mid$$[e']\cap T_I$$\mid$ $\geq$ $\mid$$[e]\cap T_I$$\mid$, $e$ is a unsuccessful cutoff (Line 8).

(2) Otherwise ($t_f$-labelled event in $[e]$), we check whether an illegal livelock exists. 
$\mathcal{A}^q_{\neg\psi}$ represent a Büchi automaton with $q$ as the initial state. 
If $\mathcal{A}_{\neg\psi}^q \models O^\omega$ means $M$ is a monitor marking (Line 12), $e$ is cutoff-\uppercase\expandafter{\romannumeral2}, and $result=false$ (Line 13).
This implies there exists an illegal livelock by Theorem \ref{The:IllTra}.

\begin{figure}[t]\centering
\includegraphics[width=0.46\textwidth]{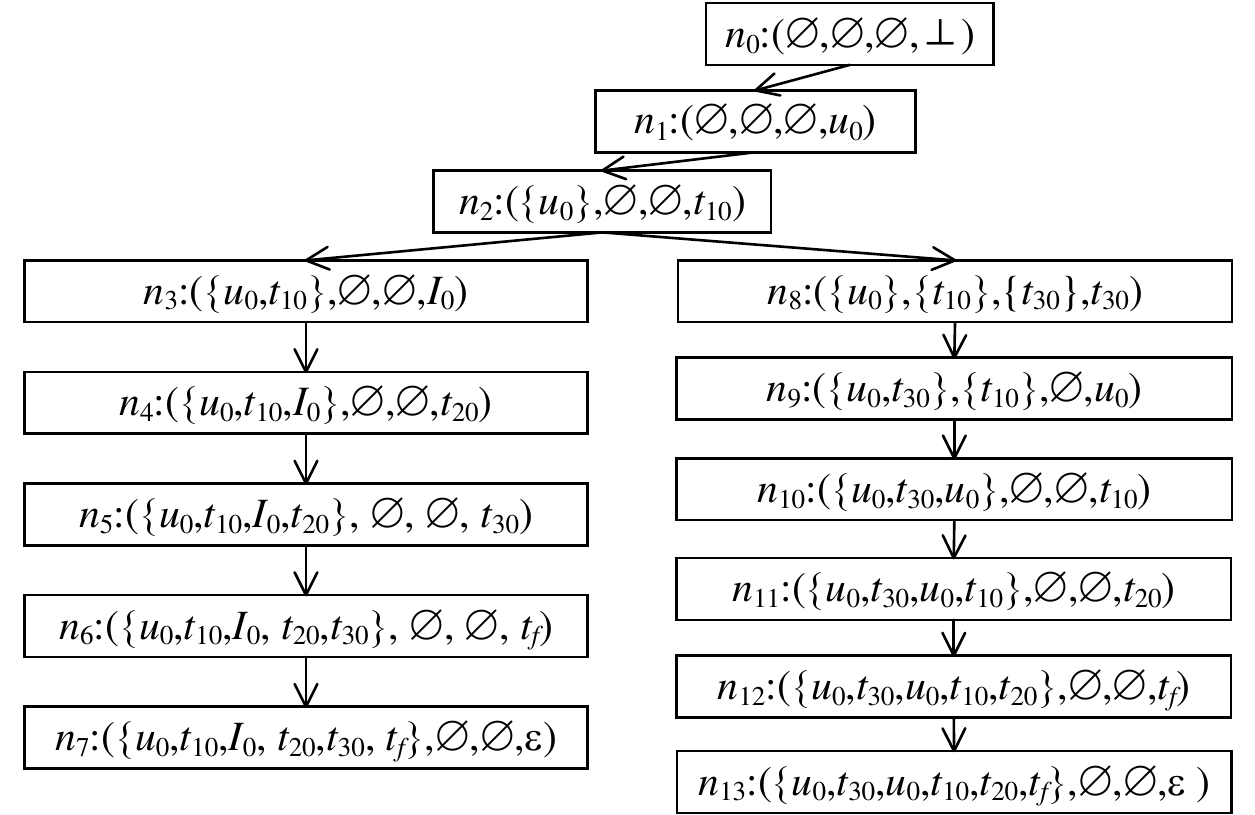}
\caption{The Exploration Tree.}\label{Fg:ExplorationTree}
\end{figure}

As we can see from our exploration tree shown in Figure \ref{Fg:ExplorationTree}, $t_{10}$ is selected to occur from $n_2$. Its right child node explores a distinct partial-order run in which $t_{10}$ is delayed. Because $t_{30}$ conflicts with $t_{10}$ according to Definition \ref{Def:Conf},
$\{t_{30}\}$ includes all causal transitions of $t_{30}$ to guide the next extension. 
Once $t_{30}$ occurs, $t_{10}$ is deleted from $\kappa$.
In fact, some actions in concurrent threads referencing the variables are modeled by the bidirectional arcs. Especially during the synchronization, the Büchi transitions observe places via the bidirectional arcs. Transitions referencing the same variable place are identified as conflict transitions without considering the concurrency of referencing actions on the same variables. 
However, when such two transitions are enabled simultaneously, their occurrence is interleaved in $Unf_{N}$ according to the traditional unfolding semantics. The resulting structure is larger in scale. This is because such conflicts lead to unnecessary interleaving of events in the unfolding structure. We also avoid such redundant structures.


\section{Experimental Evaluation}\label{Sec:Exp}

\subsection{Tool Implementation and Benchmarks}\label{Sub:Tool}

We implement an automatic tool named \textit{PUPER} (PDNet Unfolding and Partial-order explorER). 
It supports a restricted syntax of C language and POSIX threads. 
\textit{PUPER} can input a concurrent program with its LTL-$_\mathcal{X}$. 
The negation for LTL-$_\mathcal{X}$ formula is translated to a Büchi automaton and then translated to a Büchi PDNet.
Our PDNet synchronizes with the Büchi PDNet to yield a Büchi PDNet product.
The Büchi PDNet product is unfolded by constructing an exploration tree, implemented by a variant of Algorithm \ref{Alg:UPOC}.
As for checking whether the concurrent program violates the LTL-$_\mathcal{X}$ formula, it checks whether the current partial-order run has an illegal trace by checking a successful cutoff event, implemented by a variant of Algorithm \ref{Alg:Cutoff}. 

\begin{table}[h]\centering
\caption{Basic demographics of Benchmarks and Formulae. Columns are: LoC: nr. of Lines; $N_T$: nr. of threads; $\psi$: nr. of formulae; Property: concrete form; Res: Known results of each formula.}\label{Tab:Bench}
\scalebox{0.85}{
\begin{tabular}{lccclc}
\toprule[1pt]
\multicolumn{3}{c}{Benchmarks}                                            & \multicolumn{3}{c}{Formulae}                                        \\ \midrule
\multicolumn{1}{c}{Program}                    & LoC                 & $N_T$             & $\psi$   & \multicolumn{1}{c}{Property}                   & Res \\ \midrule
\multirow{1}{*}{Concur10}       & \multirow{1}{*}{92} & \multirow{1}{*}{10} & $\psi_1$ & $\mathcal{G} \neg error()$           & $F$    \\
\multirow{1}{*}{Concur20}       & \multirow{1}{*}{164} & \multirow{1}{*}{20} & $\psi_2$ & $\mathcal{G} \neg error()$           & $T$           \\     
\multirow{1}{*}{Concur30}       & \multirow{1}{*}{236} & \multirow{1}{*}{30} & $\psi_3$ & $\mathcal{G} \neg error()$           & $T$           \\  

\multirow{1}{*}{Ccnf9}      & \multirow{1}{*}{37} & \multirow{1}{*}{9} & $\psi_4$ & $\mathcal{G} \neg error()$           & $F$          \\
\multirow{1}{*}{Ccnf17}      & \multirow{1}{*}{56} & \multirow{1}{*}{17} & $\psi_5$ & $\mathcal{G} \neg error()$      & $F$           \\
\multirow{1}{*}{Ccnf19}      & \multirow{1}{*}{61} & \multirow{1}{*}{19} & $\psi_6$ & $\mathcal{G} \neg error()$      & $F$           \\

\multirow{1}{*}{Shared20}   & \multirow{1}{*}{44} & \multirow{1}{*}{20} & $\psi_7$ & $\mathcal{G} (shared<15)$           & $F$           \\
\multirow{1}{*}{Shared30}   & \multirow{1}{*}{44} & \multirow{1}{*}{30} & $\psi_8$ & $\mathcal{G}(shared<15)$           & $T$           \\
\multirow{1}{*}{Shared40}   & \multirow{1}{*}{44} & \multirow{1}{*}{40} & $\psi_9$ & $\mathcal{G} (shared<15)$           & $F$           \\
\multirow{1}{*}{Shared50}   & \multirow{1}{*}{44} & \multirow{1}{*}{50} & $\psi_{10}$ & $\mathcal{G} (shared<100)$           & $T$           \\
                           
\multirow{1}{*}{SSB10}    & \multirow{1}{*}{63} & \multirow{1}{*}{10} & $\psi_{11}$ & $\mathcal{G} ((x<2) \rightarrow \mathcal{F} (y<2))$           & $F$           \\
\multirow{1}{*}{SSB20}    & \multirow{1}{*}{63} & \multirow{1}{*}{20} & $\psi_{12}$ & $\mathcal{G} ((x<2) \rightarrow \mathcal{F} (y<2))$           & $F$           \\
\multirow{1}{*}{SSB30}    & \multirow{1}{*}{63} & \multirow{1}{*}{30} & $\psi_{13}$ & $\mathcal{G} ((x<2) \rightarrow \mathcal{F} (y<2))$           & $F$           \\

\multirow{1}{*}{Dekker}    & \multirow{1}{*}{58} & \multirow{1}{*}{2} & $\psi_{14}$ & $\mathcal{G} (flag1=1\vee flag1=0)$           & $T$           \\
                           
\multirow{1}{*}{Lamport}   & \multirow{1}{*}{81} & \multirow{1}{*}{2} & $\psi_{15}$ & $\mathcal{G} (x=0\vee x=1)$          & $F$           \\
                           
\multirow{1}{*}{Peterson}  & \multirow{1}{*}{46} & \multirow{1}{*}{2} & $\psi_{16}$ & $\mathcal{G} (flag1=1\vee flag1=0)$     & $T$           \\

\multirow{1}{*}{Szymanski} & \multirow{1}{*}{63} & \multirow{1}{*}{2} & $\psi_{17}$ & $\mathcal{G} (flag1=0 \vee flag1=1)$ & $F$           \\

\bottomrule[1pt]
\end{tabular}
}
\end{table}

Benchmarks, as outlined in Table \ref{Tab:Bench}, are adapted from the Software Verification Competition \cite{SV}, and others are used in related works \cite{Rodriguez2015Unfolding}. 
$Concur(N_T)$, $Ccnf(N_T)$, $Shared(N_T)$ and $SSB(N_T)$ \cite{Rodriguez2015Unfolding} exhibits high concurrency, where $N_T$ is the number of concurrent threads.
\textit{Lamport}, \textit{Dekker}, \textit{Szymanski}, and \textit{Peterson} \cite{SV} are four classical algorithms for solving the mutual exclusion problem in concurrent systems. 
At the same time, we set some LTL formulae for each benchmark, whose propositions specify some constraints relevant to the key variables or some specific locations. Both our source code and benchmarks used in the evaluation are available\footnote{\url{https://github.com/shirleylee03/puper/}\label{tool}}.
The experiments are conducted under the limitation of $16GB$ memory and $300s$ time. And all reported time are in seconds. 
The whole process of each verified property for each benchmark is called execution.

\subsection{Subject and Methodology}\label{Sub:Bench}

The core contribution of this paper is that we propose an on-the-fly unfolding with optimal exploration method for LTL-$_\mathcal{X}$ to avoid the enumeration of possible concurrent combinations. 
One sight is that the traditional unfolding method is used to automatically generate the unfolding structures of PNs rather than a concurrent program. Thus, we do not contrast the tools in these existing papers (Punf \cite{Esparza2001Implementing}).
Although there are unfolding-based partial-order reduction tools (POET \cite{Rodriguez2015Unfolding}) for concurrent programs, they are not used for LTL properties. 
The other insight is that we take the traditional unfolding generation as the baseline method. Here, the standard unfolding method systematically and exhaustively enumerates sets of conditions and explores all coverable sets. In this paper, we implement it with the same LTL properties.

We compare it with the traditional unfolding generation \cite{Khomenko2004Parallel}. 
This unfolding method starts with initial conditions and enumerates all combinations of conditions that can form a configuration. This process is exponential, which is the most computationally expensive step. It then determines which combinations have possible extensions and selects a possible extension to add a new event. 
It terminates until no further new event can be extended or a counterexample path that violates the LTL formula is found.
We should also compare our tool with the state-of-the-art LTL tool \textit{SPIN} and \textit{DiVine} to evaluate our performance on LTL-$_\mathcal{X}$ checking. 
Although \textit{LTSMin} can verify LTL formulae, it is not directly applicable to concurrent programs. Therefore, we do not contrast it.
Based on the above insights, we present experimental evaluations to answer the following questions to highlight two aspects:

$Q1$. How does our PDNet unfolding generation compare to the traditional unfolding generation?

$Q2$. How does our LTL-$_\mathcal{X}$ checking method compare to the existing tools?

\subsection{Comparison with Traditional Unfolding Generations}\label{Sub:ComTri}

\begin{table*}[t]\centering
\caption{Comparison with \textit{baseline} methods (in seconds). Columns are: $N_T$: the concurrent threads, $R_P$: the correct result, $|P|$: nr. of places, $|T|$: nr. of transition, $Q_{unf}$: nr. of conditions, $E_{unf}$: nr. of events, $C_{unf}$: nr. of concurrent event combinations, $T_{unf}$: the checking time of traditional unfolding generations, $Q_{pup}$: nr. of conditions, $E_{pup}$: nr. of events, $N_{pup}$: nr. of tree nodes, and $T_{pup}$: the checking time. }\label{Tab:ComTra}
\scalebox{0.8}{
\begin{tabular}{lrcrrrrrrrrrrr}
\toprule[1pt]
\multicolumn{3}{c}{Benchmarks} & \multicolumn{2}{c}{PDNet} & \multicolumn{4}{c}{\textit{baseline}}                     & \multicolumn{4}{c}{$PUPER$}                     & Comparison \\ \midrule
Name        & $N_T$   & $R_P$  & $|P|$       & $|T|$       & $Q_{unf}$ & $E_{unf}$ & $C_{unf}$ & $T_{unf}$ & $Q_{pup}$ & $E_{pup}$ & $N_{pup}$ & $T_{pup}$ & $\nabla T$ \\ \midrule
Concur10    & 10      & $F$    & 133         & 60          & 210       & 24        & 24        & 0.018     & 258       & 51        & 52        & 0.04      & 0.45       \\
Concur20    & 20      & $T$    & 248         & 106         & 1105      & 275       & 102963    & TO        & 1745      & 813       & 33565     & 64.728    & 4.64       \\
Concur30    & 30      & $T$    & 348         & 146         & 1541      & 213       & 48857     & TO        & 2576      & 975       & 52751     & 209.134   & 1.65       \\
Ccnf9       & 9       & $F$    & 91          & 41          & 147       & 22        & 22        & 0.014     & 171       & 35        & 36        & 0.029     & 0.48       \\
Ccnf17      & 17      & $F$    & 151         & 65          & 384       & 30        & 30        & 0.054     & 439       & 59        & 60        & 0.073     & 0.74       \\
Ccnf19      & 19      & $F$    & 166         & 71          & 463       & 32        & 32        & 0.065     & 526       & 65        & 66        & 0.096     & 0.68       \\
Shared20    & 20      & $F$    & 84          & 45          & 14211     & 5609      & 43290     & TO        & 480       & 128       & 129       & 0.125     & 2401.10    \\
Shared30    & 30      & $T$    & 84          & 45          & 12056     & 4713      & 42304     & TO        & 2791      & 1031      & 1452      & 3.732     & 81.54      \\
Shared40    & 40      & $F$    & 84          & 45          & 12216     & 4743      & 42334     & TO        & 920       & 248       & 249       & 0.344     & 939.96     \\
Shared50    & 50      & $T$    & 84          & 45          & 12376     & 4773      & 42364     & TO        & 7183      & 2732      & 3943      & 33.26     & 12.16      \\
SSB10       & 10      & $F$    & 74          & 37          & 467       & 143       & 195       & 0.063     & 270       & 92        & 93        & 0.048     & 1.31       \\
SSB20       & 20      & $F$    & 74          & 37          & 76        & 15        & 18        & 0.02      & 460       & 162       & 163       & 0.11      & 0.14       \\
SSB30       & 30      & $F$    & 74          & 37          & 16        & 2         & 2         & 0.011     & 650       & 232       & 233       & 0.168     & 0.07       \\
Dekker      & 2       & $T$    & 119         & 66          & 260       & 100       & 224       & 0.07      & 174       & 75        & 140       & 0.05      & 1.40       \\
Lamport     & 2       & $F$    & 156         & 91          & 50        & 10        & 10        & 0.026     & 104       & 37        & 38        & 0.018     & 1.44       \\
Peterson    & 2       & $T$    & 103         & 56          & 236       & 93        & 238       & 0.05      & 181       & 82        & 131       & 0.037     & 1.35       \\
Szymanski   & 2       & $F$    & 138         & 82          & 208       & 65        & 107       & 0.055     & 176       & 62        & 75        & 0.045     & 1.22       \\
\multicolumn{3}{c}{Average}    & 130.1       & 63.235      & 3295.412  & 1227.176  & 19000.82  & 116.3745  & 1123.765  & 404.6471  & 5480.941  & 18.35512  & 6.34
\\
\bottomrule[1pt]
\end{tabular}
}
\end{table*}

For $Q1$, we compare the performance of our method, abbreviated as $PUPER$, with the traditional unfolding generation (abbreviated as \textit{baseline}). 
To show the performance of $PUPER$ quantitatively, we calculate the reduction multiple by $\nabla T$$=$$T_{unf}/T_{pup}$ in Table \ref{Tab:ComTra}.
Specifically, TO means that the checking time exceeds the upper time limitation ($300s$) and does not give a deterministic result. 

As we can see from Table \ref{Tab:ComTra}, except for six executions from \textit{baseline} where the output is TO, all of them get the same verification results as $R_P$. 
This implies that our methods are correct. 
As for the checking time, $T_{unf}$$>$$T_{pup}$ holds for $11$ execution, while $T_{pup}$$>$$T_{unf}$ holds for the rest of executions. 
Although \textit{baseline} can reduce unnecessary interleaving compared with the reachability graph, it still needs to enumerate all possible concurrent event combinations ($C_{unf}$) during the calculation of the next event, which brings a high cost. 
If the number of event combinations $C_{unf}$ is large (reaching the scale of $10^4$), the checking results cannot be obtained within a limited time.

In $Concur$ (for the same LTL formula) benchmarks, as $N_T$ increases, the advantage of $PUPER$ is clear. \textit{Baseline} has timed out within $20$ threads, but $PUPER$ can output the results within $300s$ within $30$ threads. 

In $Shared$ (for the same LTL formula) benchmarks, 
\textit{baseline} has timed out due to the large-scale number of event combinations $C_{unf}$, but $PUPER$ can output the results within $300s$ within $50$ threads. 

In $Ccnf$ and $SSB$ (for some different LTL formulae) benchmarks, although $PUPER$ outperforms \textit{baseline} under $10$ threads of $SSB$, \textit{baseline} can also exhibit good performance under other benchmarks. 
The reason is that $C_{unf}$ of \textit{baseline} is small (not reaching the scale of $10^4$), without bringing huge costs to check the first illegal trace.

In the four classical mutual exclusion algorithms with only $2$ threads, \textit{baseline} can also show good performance. Although the construction of the exploration tree brings additional overhead, the checking time $T_{unf}$ and $T_{pup}$ are on the same order of magnitude. The overhead is still acceptable.

In conclusion, $PUPER$ proposed in this paper has a great optimization effect in mitigating the NP-complete problem by avoiding enumerating
possible concurrent event combinations. TO means that no result was obtained within the limited $300s$. These cases in $Concur$ and $Shared$ reflect too many possible combinations. Traditional unfolding methods traverse all possible combinations when calculating possible extensions. This may reflect the effectiveness of our method. 
Even in the case where there are not so many combinations from \textit{baseline}, $PUPER$ also does not bring too much overhead.
$PUPER$ spends $6.34$ times less time than \textit{baseline} on average. 

\subsection{Comparison with LTL Tools}\label{Sub:ComOth}

\begin{table}[t]\centering
\caption{Comparison with \textit{SPIN} and \textit{DiVine} (in seconds). Columns are: $N_{div}$: nr. of explored nodes, $T_{div}$: time of \textit{DiVine}; $N_{spin}$: nr. of the explored nodes, $T_{spin}$: time of \textit{SPIN}; $N_{pup}$: nr. of the explored nodes, and $T_{pup}$: time of \textit{PUPER}.}\label{Tab:ComOth}
\scalebox{0.7}{
\begin{tabular}{crrrrrrrr}
\toprule[1pt]
\multirow{2}{*}{$N_T$} & \multicolumn{2}{c}{$DiVine$} & \multicolumn{2}{c}{$SPIN$} & \multicolumn{2}{c}{$PUPER$} & \multicolumn{2}{c}{Comparison} \\ \cmidrule{2-9}
& $N_{div}$     & $T_{div}$    & $N_{spin}$   & $T_{spin}$  & $N_{pup}$    & $T_{pup}$    & $\nabla d$     & $\nabla s$    \\ \midrule
3                      & 115           & 0.066        & 423          & 0           & 25           & 0.017        & 3.88           & 0.00          \\
4                      & 243           & 0.103        & 2348         & 0           & 30           & 0.024        & 4.29           & 0.00          \\
5                      & 499           & 0.242        & 13098        & 0.02        & 35           & 0.032        & 7.56           & 0.63          \\
6                      & 1011          & 0.372        & 72473        & 0.16        & 40           & 0.042        & 8.86           & 3.81          \\
7                      & 2035          & 0.704        & 397437       & 0.97        & 45           & 0.055        & 12.80          & 17.64         \\
8                      & 4083          & 1.356        & 2162169      & 6.32        & 50           & 0.067        & 20.24          & 94.33         \\
9                      & 8179          & 2.878        & 11666054     & 37.7        & 55           & 0.084        & 34.26          & 448.81        \\
10                     & 16371         & 6.708        & 62045088     & 226         & 60           & 0.095        & 70.61          & 2378.95       \\
Average                & 4067          & 1.55         & 9544886.25   & 33.90       & 42.50        & 0.05         & 20.31          & 368.02
\\
\bottomrule[1pt]
\end{tabular}
}
\end{table}

Our \textit{PUPER} can verify whether the concurrent program satisfies the LTL-$_\mathcal{X}$ formula.
For $Q2$, since the tools in \cite{Rodriguez2015Unfolding,Coti2021Quasi,Schemmel2020Symbolic} are designed for safety properties instead of LTL properties, we compare \textit{PUPER} with \textit{DiVine} \cite{Barnat2013DiVinE} and \textit{SPIN} \cite{Holzmann2002An}  two state-of-the-art tools for concurrent programs on LTL properties. 
The benchmarks used here are $Concur$ with the thread number from $3$ to $10$, which exhibit high concurrency. Their results are all $T$ to compare the checking better in the case of complete nodes. 
They all use partial order reduction techniques to alleviate the state space explosion problem. Such techniques are also based on unconditional dependencies and can help reduce the number of paths that need to be explored. Therefore, they are well suited for comparison with the on-the-fly unfolding with optimal exploration method in this paper.

In Table \ref{Tab:ComOth}, $N_{div}$, $T_{div}$, $N_{spin}$ and $T_{spin}$ are output by themselves.
As we can see from Table \ref{Tab:ComOth}, the total time of our tool grows the slowest as the number of threads increases. 
Overall, the average of $T_{div}$ is $1.55s$, the average of $T_{spin}$ is $33.9s$, and the average of \textit{PUPER} of $T_{pup}$ is only  $0.05s$.
To show the effect of \textit{PUPER} quantitatively, we calculate the reduction multiple by $\nabla d$$=$$T_{div}/T_{pup}$ and $\nabla s$$=$$T_{spin}/T_{pup}$.
Thus, \textit{PUPER} spends $20.31$ times less time than \textit{DiVine} tool and $368.02$ times less time than \textit{SPIN} tool on the benchmarks. 

In summary, the fewer conflict transitions and the fewer branching paths the exploration tree needs to explore, the better our method performs.

\section{Related works}\label{Sec:Rel}

Net unfolding techniques \cite{Wallner1998Model,Esparza2001Implementing,Khomenko2004Parallel} can be employed to verify LTL properties of concurrent systems. Esparza et al. \cite{Esparza2001Implementing} propose a new unfolding approach that can solve the model-checking problem by directly checking a finite prefix without running complex algorithms on the prefix. Khomenko et al. \cite{Khomenko2004Parallel} propose an unfolding-based approach to LTL-$_\mathcal{X}$ model checking of high-level PNs based on \cite{Esparza2001Implementing} for low-level nets. These methods represent the traditional unfolding generation (\textit{baseline}) in this paper. 
And most unfolding-based model-checking methods should construct a complete structure for property verification in advance.

As for concurrent programs, verification methods \cite{Rodriguez2015Unfolding,Coti2021Quasi,Schemmel2020Symbolic} based on event structure (a simplified structure using unfolding semantics) combine the dynamic partial-order reduction \cite{Abdulla2014Optimal}. The event structure is constructed based on the unconditional independence of the execution model, and the optimal exploration can cover all reachable states. Some testing methods \cite{Kahkonen2014Lightweight,Kahkonen2017Testing,Saarikivi2017Minimizing} are combined with dynamic symbolic execution, where the entire execution paths of concurrent programs are directly expressed by the unfolding structure of plain PNs or contextual nets (a kind of PNs with read arcs) to explore the target states. 
Recently, the Verification of concurrent programs using Petri net unfolding \cite{Dietsch2021Verification} can verify safety properties within the framework of counterexample-guided abstraction and refinement. 
They propose new unfolding methods to solve state-explosion problems for concurrent programs. However, these methods are not applicable for LTL properties.

In the case of LTL of software model checking, Dietsch et al. \cite{Dietsch2015Fairness} propose a new approach to LTL properties for sequential programs based on fairness modulo theory. However, this method only applies to sequential programs, not the concurrent programs targeted in this paper.

\section{Conclusion}\label{Sec:Con}
The on-the-fly method for LTL model checking is not effective for concurrent interleaving. Constructing a complete unfolding structure in advance and then checking LTL is also wasteful. It is difficult to apply a tree-like structure directly to the traditional on-the-fly unfolding. 
We propose a novel method called on-the-fly unfolding with optimal exploration for LTL. The traditional conflict relations between transitions in PNs are insufficient to cover all the conflict actions in PDNet-modeled concurrent programs. Based on dependencies from PDNet, we propose a formal definition of conflict transition to identify conflict relations.
The concept of the disabled set cannot ensure the entire partial-order runs, so we define an exploration tree with a novel notion of a delayed set of transitions. Then, we propose an LTL-$_\mathcal{X}$ algorithm to simultaneously generate the unfolding structure while constructing our exploration tree.
Finally, we develop a tool for concurrent programs and validate that our methods are compared with traditional unfolding methods and existing LTL tools. 

In the future, we will propose a merged process as a more compact unfolding structure to represent all partial-order runs.

\section*{Acknowledgments}\label{Sec:Ack}

This work was supported by National Key Research and Development Program of China under Grant 2022YFB4501704 and National Natural Science Foundation of China under Grant 62302306.



\printcredits

\bibliographystyle{cas-model2-names}
\bibliography{cas-refs-puper}

\newpage

\appendix 
\section*{Appendix}

\begin{subappendices}

\renewcommand{\thesection}{\Alph{section}}%

\section{PDNet Synchronization}\label{App:Syn}

\setcounter{table}{0}   
\setcounter{figure}{0}
\setcounter{equation}{0}
\setcounter{algorithm}{0}
\setcounter{definition}{0}
\setcounter{theorem}{0}
\setcounter{remark}{0}

\renewcommand{\thetable}{A.\arabic{table}}
\renewcommand{\thefigure}{A.\arabic{figure}}
\renewcommand{\theequation}{A.\arabic{equation}}
\renewcommand{\thealgorithm}{A.\arabic{algorithm}}
\renewcommand{\thedefinition}{A.\arabic{definition}}
\renewcommand{\thetheorem}{A.\arabic{theorem}}
\renewcommand{\theremark}{A.\arabic{remark}}

Inspired by the automata-theoretic approach with on-the-fly exploration \cite{He2021More}, the key to verifying LTL properties by unfolding is to synchronize PDNet with a Büchi automaton accepting the negation of LTL-$_\mathcal{X}$.
The model-checking problem is reduced to the emptiness checking of illegal traces by traditional unfolding generation \cite{Esparza2002An}.
We improve the synchronization \cite{Esparza2001Implementing,Khomenko2004Parallel} to PDNet.

Let $\psi$ be an LTL-$_\mathcal{X}$ formula. For any proposition $po$ in $\psi$, $[po]$ represents a boolean expression corresponding to this proposition. 
Firstly, a Büchi automaton $\mathcal{A}_{\neg\psi}$ accepting $\neg\psi$ is translated efficiently \cite{He2021More}. 
To yield a synchronization, $\mathcal{A}_{\neg\psi}$ should be translated into an equivalent PDNet.
It is a PDNet with acceptance capabilities called Büchi PDNet. 
Formally, Büchi PDNet is $N_\mathcal{A}::=(N,P_{A})$, where $N$ is a PDNet as Definition \ref{Def:PDNet} and $P_{A}\subseteq P$ is a finite set of acceptable places (corresponding to the acceptable state of the Büchi automaton).

Let $N_{\neg\psi}$ be a Büchi PDNet translated from $\mathcal{A}_{\neg\psi}$, such that

(1) Each state $q$ of $\mathcal{A}_{\neg\psi}$ is translated to a place $p_q$ (called Büchi place, denoted by $P_\mathcal{A}$), where only the place for the initial state $q_0$ has a token. 

(2) Each transition relation $(q,po,q')$ of $\mathcal{A}_{\neg\psi}$ is translated to a transition $t_x$ (called Büchi transition, denoted by $T_\mathcal{A}$), where $p_q$ and $p_{q'}$ are the input and output places of $t_x$, and $G(t_x)=[po]$ is the guard expression for $t_x$. 

(3) Each acceptable state of $\mathcal{A}_{\neg\psi}$ is translated to an acceptable place $p_a$ belonging to $P_{A}$.

For any word on a run $\varpi$ accepted by $\mathcal{A}_{\neg\psi}$, there exists an occurrence sequence $\pi$ (defined in \cite{2023arXiv2301Ding}) corresponding to $\varpi$ accepted by $N_{\neg\psi}$. 
An occurrence sequence (run) $\pi$ is accepted by $N_\mathcal{A}$ if $\exists t\in ^\bullet$$P_{A}$ appears infinitely often in $\pi$.
We define a distinguished subset of Büchi transitions for $P_{A}$ (called $I$-transition), by $T_I ::= \{t\mid \forall t\in T_\mathcal{A}\colon t$$^\bullet \cap P_{A}\neq\emptyset\}$.
Thus, the acceptable occurrence sequence occurs some $I$-transitions infinitely often. 
For synchronization, Büchi transition $t\in T_\mathcal{A}$ should observe some necessary places in $N$ to evaluate the truth of guard expression $G(t)$. We define them as observable places of PDNet from LTL-$_\mathcal{X}$.

\begin{definition}[Observable Place of PDNet]\label{Def:Obs}
Let $\psi$ be a LTL-$_\mathcal{X}$ of $N$, $Po$ all propositions from $\psi$, and $po\in Po$ a proposition.
	
If a proposition $po$ is in the form of $is-fireable(t)$, $Obs([po])::=\{p\mid p\in ^\bullet$$t\cap P_f \vee (p\in ^\bullet$$t\cap P_v \wedge Var(G(t))\neq\emptyset)\}$. 	
If a proposition $po$ is in the form of $token-value(p)\emph{ }rop\emph{ }c$, $Obs([po])::=\{p\}$.
The set of observable places over $\psi$ is $P_{Obs} ::= \{p\mid \forall po\in Po\colon p\in Obs([po])\}$.
\end{definition}

Thus, an observable place is either an execution place from $P_f$ or a variable place from $P_v$.
Initially, a standard approach synchronizes $N$ with $N_{\neg\psi}$ on all transitions \cite{Wallner1998Model}.
However, such a synchronization shall disable all concurrency.
It contradicts the essence of unfolding, exploiting the concurrency to generate a compact representation of the state space. 
Thus, we define visible transitions of PDNet. Intuitively, such transitions change the making of some observable places.

\begin{definition}[Visible Transition of PDNet]\label{Def:Vis}
Let $\psi$ be an LTL-$_\mathcal{X}$ formula w.r.t. $N$, and $P_{Obs}$ the set of observation places w.r.t. $\psi$. 
The set of visible transitions is
$T_{Vis} ::= \{t \mid\forall t\in T,\forall p\in P_{Obs} \cap P_v\colon t\in ^\bullet$$p \wedge E(p,t) \neq E(t,p) \vee \forall p\in P_{Obs}\cap P_f\colon  t \in ^\bullet$$p \cup p$$^\bullet\}$.
\end{definition}

Obviously, $T_{Inv}=T\setminus T_{Vis}$ is called invisible transitions. 
As the example in Figure \ref{Fg:Motivation}($i$), a Büchi PDNet is translated from the Büchi automaton in Figure \ref{Fg:Motivation}($c$). 
Here, $u_0$ corresponds to $(q_0,true,q_0)$, $I_0$ corresponds to $(q_0,(x=1)\cap (z\neq 1),q_1)$, $I_1$ corresponds to $(q_1,(z\neq 1),q_1)$. 
Thus, $G(I_0)$ is $[v_x =1\cap v_z\neq 1]$, and $G(I_1)$ is $[v_z\neq 1]$. 
Then, $T_I=\{I_0,I_1\}$, $P_{Obs}=\{v_x,v_z\}$,  $T_{Vis}=\{t_{10},t_{30}\}$ and $T_{Inv}=\{t_{20}\}$.

Based on the above definitions, we propose a synchronization of PDNet (called Büchi PDNet product) in Algorithm \ref{Alg:Syn}.
Let $N$ be a PDNet, $\psi$ be an LTL-$_\mathcal{X}$ of $N$, and $N_{\neg\psi}$ be a Büchi PDNet of $\neg\psi$. $\mathcal{N}_{\neg\psi}$ represents a Büchi PDNet product of $N$ and $N_{\neg\psi}$.
Intuitively, $\mathcal{N}_{\neg\psi}$ can be seen as a Büchi PDNet by placing $N$ in a suitable environment $N_{\neg\psi}$.
$P_{A}$ and $T_I$ of $\mathcal{N}_{\neg\psi}$ are the same as $N_{\neg\psi}$.
Similarly, a run accepted by $\mathcal{N}_{\neg\psi}$ enabled from the initial marking occurs some $I$-transitions infinitely often.
The existence of such acceptable run (called illegal infinite-trace) of $\mathcal{N}_{\neg\psi}$ implies that $N$ (the corresponding concurrent program $\mathcal{P}$) violates $\psi$.

\begin{breakablealgorithm}
	\caption{Synchronization Algorithm}\label{Alg:Syn}
		\begin{algorithmic}[1]
			\Require A PDNet $N$ and an LTL-$_\mathcal{X}$ formula $\psi$;
			\Ensure Büchi PDNet Product $\mathcal{N}_{\neg\psi}$;
			\State $T_I := \emptyset$;
			\State $P:=P\cup \{P_B, P_S\}$; /$\ast I(P_B)\neq\emptyset\ast$/
			\ForAll {$t\in T_\mathcal{A}$} 
			\If {$G(t) \neq\emptyset$} 
			\ForAll {$p\in P_{Obs}\cap Obs(G(t))$} 
			\State $F:= F\cup\{(t,p),(p,t)\}$;
			\EndFor
			\EndIf
			\If{$t$$^\bullet\cap P_A\neq\emptyset$}
			\State $T_I := T_I \cup \{t\}$;
			\EndIf
			\State $F:=F\cup\{(P_B, t),(t, P_S)\}$;
			\EndFor
			
			\ForAll {$t\in T_{Vis}$} 
			\State $F:=F\cup\{(P_S, t),(t, P_B)\}$; 
			\EndFor
			\State Let $t_f$ be a new transition; 
			\ForAll {$p \in P \wedge p^\bullet =\emptyset$}
			\State $F:=F\cup\{(p, t_f),(t_f, p)\}$;
			\EndFor
			
		\end{algorithmic}
\end{breakablealgorithm}

In order to make the alternate execution of $N$ and $N_{\neg\psi}$, the scheduling places $P_B$ and $P_S$ are added to $\mathcal{N}_{\neg\psi}$.
The intuition behind these places is that when $P_B$ ($P_S$) has a token, it is the turn of $N_{\neg\psi}$ ($N$) to occur an enabled transition.
Because Büchi PDNet $N_{\neg\psi}$ needs to observe the initial marking of $N$, only $P_B$ has a token initially (Line 2 in Algorithm \ref{Alg:Syn}).
Intuitively, Büchi transitions of $\mathcal{N}_{\neg\psi}$ should touch the observable places. Based on Definition \ref{Def:Obs}, the observed places are determined by $Obs(G(t))$, where the guard expression $G(t)$ corresponds to a proposition (Line 5).
Then, $p$ is observed by $t$ via bidirectional arcs (Line 6).
At the same time, $I$-transitions are identified by $T_I$ (Lines 7-8).
The visible transitions are synchronized with Büchi transitions by scheduling places.
Here, Büchi transition in $T_\mathcal{A}$ is scheduled by the two arcs (Line 9), and visible transitions in $T_{Vis}$ are scheduled by the two arcs (Line 11).
Since LTL-$_\mathcal{X}$ model-checking is based on infinite sequences, we construct a transition $t_f$ to build the infinite sequences (Lines 13-14) by a self-loop.

As the example in Figure \ref{Fg:Motivation}($j$), we show how to construct the synchronization of $N_{\neg\psi}$ and $N$, where bolded arcs are added by Algorithm \ref{Alg:Syn}. 
For instance, $(v_x, I_0)$ and $(I_0,v_x)$ represent the observation of Büchi transition $I_0$. $(P_B,I_0)$ and $(I_0,P_S)$ represent the turn from $N_{\neg\psi}$ to $N$. 

\section{Proof}\label{App:Proof}

\setcounter{table}{0}   
\setcounter{figure}{0}
\setcounter{equation}{0}
\setcounter{algorithm}{0}
\setcounter{definition}{0}
\setcounter{theorem}{0}
\setcounter{remark}{0}

\renewcommand{\thetable}{B.\arabic{table}}
\renewcommand{\thefigure}{B.\arabic{figure}}
\renewcommand{\theequation}{B.\arabic{equation}}
\renewcommand{\thealgorithm}{B.\arabic{algorithm}}
\renewcommand{\thedefinition}{B.\arabic{definition}}
\renewcommand{\thetheorem}{B.\arabic{theorem}}
\renewcommand{\theremark}{B.\arabic{remark}}

To prove the correctness, we split Theorem \ref{The:IllTra} into some sub-theorems.

\begin{theorem}\label{The:cutoff1-1}
	Let $N$ be a PDNet, $\psi$ be an LTL-$_\mathcal{X}$ formula, $\mathcal{N}_{\neg\psi}$ be the Büchi PDNet product, and $\hbar$ be a prefix of $Unf_\mathcal{N}$.
	If $\hbar$ contains a successful cutoff-\uppercase\expandafter{\romannumeral1}, then $\mathcal{N}_{\neg\psi}$ has an illegal infinite-trace.
\end{theorem}
\begin{proof}\label{Pro:cutoff1-1} 
Let $e$ be a successful cutoff-\uppercase\expandafter{\romannumeral1}, whose companion is $e'$.
In particular, $e'<e$ implies $[e']\prec[e]$.
Let $\pi_1\pi_2$ be an occurrence sequence of $\mathcal{N}_{\neg\psi}$, such that $\pi_1$ is a linearization of $[e']$ and $\pi_1\pi_2$ is a linearization of $[e]$. Let $M_1$ be the marking reached by $\pi_1$, $M_2$ be the reached by $\pi_1\pi_2$. 
Since $Mark([e'])=Mark([e])$, $M_1=M_2$.
According to Definition 9, $[e']\setminus[e]\cap T_I\neq\emptyset$.
Thus, $\pi_1(\pi_2)^\omega$ is an infinite sequence of $\mathcal{N}_{\neg\psi}$ with infinitely many occurrences of $I$-transitions.
It is an illegal infinite-trace. 
\end{proof}

\begin{definition}\label{Def:Bad}
A configuration $\mathcal{C}$ of the unfolding of $\mathcal{N}_{\neg\psi}$ is bad if it contains at least $S+1$ $I-events$, where $S$ is the number of reachable markings of $\mathcal{N}_{\neg\psi}$.
\end{definition}

\begin{lemma}\label{Lem:Bad}
    $\mathcal{N}_{\neg\psi}$ has an illegal infinite-trace iff $Unf_\mathcal{N}$ contains a bad configuration.
\end{lemma}
\begin{proof}\label{Pro:Bad} 
($\Rightarrow$) Let $\pi$ be a prefix of an illegal infinite-trace of $\mathcal{N}_{\neg\psi}$, such that $\pi$ contains $S+1$ $I-transitions$. 
There exists a configuration $\mathcal{C}$ such that $\pi$ is a linearization of $\mathcal{C}$. Thus, $\mathcal{C}$ is a bad configuration.

($\Leftarrow$) Let $\mathcal{C}$ be a bad configuration. $\mathcal{C}$ contains at least $S+1$ $I-events$. 
A lemma that no reachable marking of $\mathcal{N}_{\neg\psi}$ concurrently enables two different $I$-transitions has been proved \cite{Esparza2001Implementing}.
These events are causally ordered, and by the pigeonhole principle $\exists e',e\in\mathcal{C}\colon e'<e \wedge Mark([e'])=Mark([e])$.
Let $\pi_1\pi_2$ be an occurrence sequence of $\mathcal{N}_{\neg\psi}$, such that $\pi_1$ is a linearization of $[e']$ and $\pi_1\pi_2$ is a linearization of $[e']$. Let $M_1$ be the marking reached by $\pi_1$, $M_2$ be the reached by $\pi_1\pi_2$. 
Since $Mark([e'])=Mark([e])$, $M_1=M_2$.
Thus, $\pi_1(\pi_2)^\omega$ is an illegal infinite-trace of $\mathcal{N}_{\neg\psi}$ with infinitely many occurrences of $I$-transitions.
\end{proof}

\begin{lemma}\label{Lem:BadCutoff}
    A bad configuration contains at least a successful cutoff-\uppercase\expandafter{\romannumeral1}.
\end{lemma}
\begin{proof}\label{Pro:BadCutoff}
According to Proof \ref{Pro:Bad}($\Leftarrow$), event $e$ is a successful cutoff-\uppercase\expandafter{\romannumeral1}.
\end{proof}

\begin{theorem}\label{The:Cutoff1-2}
	Let $N$ be a PDNet, $\psi$ be an LTL-$_\mathcal{X}$ formula, $\mathcal{N}_{\neg\psi}$ be the Büchi PDNet product, and $\hbar$ be a prefix of $Unf_\mathcal{N}$.
	If $\mathcal{N}_{\neg\psi}$ has an illegal infinite-trace, then $\hbar$ contains a successful cutoff-\uppercase\expandafter{\romannumeral1}.
\end{theorem}

\begin{proof}\label{Pro:Cutoff1-2}
By Lemma \ref{Lem:BadCutoff}, it suffices to show that if the unfolding of $\mathcal{N}_{\neg\psi}$ contains a bad configuration, then $\hbar$ has a successful cutoff-\uppercase\expandafter{\romannumeral1}. 
Based on Lemma \ref{Lem:Bad}, $\mathcal{N}_{\neg\psi}$ has an illegal infinite-trace, $\hbar$ has a bad configuration. Then, it also has a successful cutoff-\uppercase\expandafter{\romannumeral1}.
\end{proof}

\begin{theorem}\label{The:Cutoff2-1}
	Let $N$ be a PDNet, $\psi$ be an LTL-$_\mathcal{X}$ formula, $\mathcal{N}_{\neg\psi}$ be the Büchi PDNet product, and $\hbar$ be a prefix of $Unf_\mathcal{N}$.
	If $\hbar$ contains a successful cutoff-\uppercase\expandafter{\romannumeral2}, then $\mathcal{N}_{\neg\psi}$ has an illegal livelocks.
\end{theorem}
\begin{proof}\label{Pro:Cutoff2-1}
Let $e$ be a successful cutoff-\uppercase\expandafter{\romannumeral2}, whose companion is $e'$.
Thus, $\exists e_f\in [e]\colon\lambda(e_f)=t_f$ and $Mark([e'])$ is a monitor marking. 
Let $\pi_1\pi_2\pi_3$ be an occurrence sequence of $\mathcal{N}_{\neg\psi}$, such that $\pi_1$ is a linearization of $[e']$, $\pi_1\pi_2$ is a linearization of $[e_f]$ and $\pi_1\pi_2\pi_3$ is a linearization of $[e]$. Let $\mathcal{A}_{\neg\psi}$ be the Büchi automaton of $\neg\psi$, and $q_n$ be the Büchi projection under $Mark([e'])$.
Due to $Mark([e'])$ as a monitor marking, $\mathcal{A}_{\neg\psi}^n$ ($\mathcal{A}_{\neg\psi}$ with the initial state $q_n$) can accept such infinite marking projection sequence $(\downarrow Mark([e']))^\omega$.
Thus, $\pi_2\pi_3$ contains only invisible transitions. Let $M_2$ be the marking reached by $\pi_1\pi_2$, $M_3$ be the reached by $\pi_1\pi_2\pi_3$. 
According to Algorithm 1, $t_f$ is added for infinite occurrence sequence but does not change the marking, implying $Mark([e])=Mark([e_f])$. 
Since $Mark([e])=Mark([e_f])$, $M_2=M_3$. 
Thus, $\pi_1\pi_2(\pi_3)^\omega$ is an illegal livelock of $\mathcal{N}_{\neg\psi}$.
\end{proof}

\begin{theorem}\label{The:Cutoff2-2}
	Let $N$ be a PDNet, $\psi$ be an LTL-$_\mathcal{X}$ formula, $\mathcal{N}_{\neg\psi}$ be the Büchi PDNet product, and $\hbar$ be a prefix of $Unf_\mathcal{N}$.
	If $\mathcal{N}_{\neg\psi}$ has an illegal livelocks, then $\hbar$ contains a successful cutoff-\uppercase\expandafter{\romannumeral2}.
\end{theorem}
\begin{proof}\label{Pro:Pro:Cutoff2-2}
Let $\pi_1\pi_2(\pi_3)^\omega$ be an illegal livelocks of $\mathcal{N}_{\neg\psi}$ such that $\pi_2\pi_3$ contains no visible transitions. 
Based on Algorithm 1, there must exist an event $e_f$ satisfying $\lambda(e_f)=t_f$.
Due to the feature of $t_f$, let $\pi_1\pi_2$ be a linearization of $[e_f]$. 
Let $e$ be an event such that $\pi_1\pi_2\pi_3$ is a linearization of $[e]$ and $e_f\in [e]$.
Then, let $e'$ be the companion event of $e$ such that $\pi_1$ is a linearization of $[e']$. 
Because $\pi_2\pi_3$ contains only invisible transitions, $Mark([e'])$ is a monitor marking.
It implies that $Mark([e'])=Mark([e])$. 
That is, $e$ is a successful cutoff-\uppercase\expandafter{\romannumeral2}.
Thus, $\hbar$ has a successful cutoff-\uppercase\expandafter{\romannumeral2}. 
\end{proof}

\section{Successor Relation}\label{App:Succ}

\setcounter{table}{0}   
\setcounter{figure}{0}
\setcounter{equation}{0}
\setcounter{algorithm}{0}
\setcounter{definition}{0}
\setcounter{theorem}{0}
\setcounter{remark}{0}

\renewcommand{\thetable}{C.\arabic{table}}
\renewcommand{\thefigure}{C.\arabic{figure}}
\renewcommand{\theequation}{C.\arabic{equation}}
\renewcommand{\thealgorithm}{C.\arabic{algorithm}}
\renewcommand{\thedefinition}{C.\arabic{definition}}
\renewcommand{\thetheorem}{C.\arabic{theorem}}
\renewcommand{\theremark}{C.\arabic{remark}}

Let $s::=\langle h, m, r, u\rangle$ be a state of $\mathcal{P}$, where $h$ is a function that indicates the current program location of every thread, $m$ is the current memory state, $r$ is a function that maps every mutex to a thread identifier, and $u$ is also a function maps every condition variable to a multiset of thread identifiers of those threads that currently wait on that condition variable.
The successor relation in Table \ref{Tab:Suc} is represented by $s$$\stackrel{\langle\tau,\beta\rangle}{\longrightarrow}$$s'$, updating $s$ to a new state $s'$.

\begin{table}[ht]
\caption{Successor relation on actions.}\label{Tab:Suc}
	\begin{center}
		\begin{tabular}{l}
			\toprule[1pt]
			
			$\frac{\tau := \langle i, q, l, l', m, m'\rangle\in\mathcal{T}\emph{ }q :=\nu:=w\emph{ }h(i)=l}
			{\langle h, m, r, u\rangle
				\stackrel{\langle \tau, \langle ass, l'\rangle\rangle}{\longrightarrow}
				\langle h[i\mapsto l'], m', r, u\rangle}$  ($\rho_{ass}$) \\
			
			
			
			$\frac{\tau := \langle i, q, l, l', m, m'\rangle\in\mathcal{T}\emph{ }q := if(w)or while(w)\emph{ }[\![w]\!]m=true\emph{ }h(i)=l}
			{\langle h, m, r, u\rangle
				\stackrel{\langle \tau, \langle tcd, l'\rangle\rangle}{\longrightarrow}
				\langle h[i\mapsto l'], m, r, u\rangle}$ ($\rho_{tcd}$) \\

			$\frac{\tau := \langle i, q, l, l', m, m'\rangle\in\mathcal{T}\emph{ }q := if(w) or while(w) \emph{ } [\![w]\!]m=false \emph{ }h(i)=l}
			{\langle h, m, r, u\rangle
				\stackrel{\langle \tau, \langle fcd, l'\rangle\rangle}{\longrightarrow}
				\langle h[i\mapsto l'], m, r, u\rangle}$ ($\rho_{fcd}$) \\


			
			$\frac{\tau := \langle i, q, l, l', m, m'\rangle\in\mathcal{T}\emph{ }q := \langle lock, \ell\rangle\emph{ }h(i)=l\emph{ }r(\ell)=0}
			{\langle h, m, r, u\rangle
				\stackrel{\langle \tau, \langle acq, \ell\rangle\rangle}{\longrightarrow}
				\langle h[i\mapsto l'], m, r[\ell\mapsto i], u\rangle}$ ($\rho_{acq}$) \\
			
			$\frac{\tau := \langle i, q, l, l', m, m'\rangle\in\mathcal{T}\emph{ }q := \langle unlock, \ell\rangle\emph{ }h(i)=l\emph{ }r(\ell)=i}
			{\langle h, m, r, u\rangle
				\stackrel{\langle \tau, \langle rel, \ell\rangle\rangle}{\longrightarrow}
				\langle h[i\mapsto l'], m, r[\ell\mapsto 0], u\rangle}$ ($\rho_{rel}$) \\
			
			$\frac{\tau := \langle i, q, l, l', m, m'\rangle\in\mathcal{T}\emph{ }q := \langle signal, \gamma\rangle\emph{ }h(i)=l\emph{ }\{j\}\in u(\gamma)}
			{\langle h, m, r, u\rangle
				\stackrel{\langle \tau, \langle sig, \gamma\rangle\rangle}{\longrightarrow}
				\langle h[i\mapsto l'], m, r, u[\gamma\mapsto u(\gamma)\setminus\{j\}\cup\{-j\}]\rangle}$ ($\rho_{sig}$) \\
			
			$\frac{\tau := \langle i, q, l, l', m, m'\rangle\in\mathcal{T}\emph{ }\emph{ }q := \langle wait, \gamma, \ell\rangle\emph{ }h(i)=l\emph{ }r(\ell)=i\emph{ }\{i\}\notin u(\gamma)}
			{\langle h, m, r, u\rangle
				\stackrel{\langle \tau, \langle wa_1, \gamma, \ell\rangle\rangle}{\longrightarrow}
				\langle h, m, r[\ell\mapsto 0], u[\gamma\mapsto u(\gamma)\cup\{i\}]\rangle}$ ($\rho_{wa_1}$) \\
			
			$\frac{\tau := \langle i, q, l, l', m, m'\rangle\in\mathcal{T}\emph{ }q := \langle wait, \gamma, \ell\rangle\emph{ }h(i)=l\emph{ }r(\ell)=0\emph{ }\{-i\}\in u(\gamma)}
			{\langle h, m, r, u\rangle
				\stackrel{\langle\tau,\langle wa_2,\gamma, \ell\rangle\rangle}{\longrightarrow}
				\langle h, m, r, u[\gamma\mapsto u(\gamma)\setminus\{-i\}]\rangle}$ ($\rho_{wa_2}$)
			\\
			
			$\frac{\tau := \langle i, q, l, l', m, m'\rangle\in\mathcal{T}\emph{ }q := \langle wait,\gamma,\ell\rangle\emph{ }h(i)=l\emph{ }r(\ell)=0}
			{\langle h, m, r, u\rangle
				\stackrel{\langle\tau,\langle wa_3,\gamma, \ell\rangle\rangle}{\longrightarrow}
				\langle h[i\mapsto l'], m, r[\ell\mapsto i], u\rangle}$($\rho_{wa_3}$) \\
			\bottomrule[1pt]
		\end{tabular}
	\end{center}
\end{table}

For $\rho_{ass}$, $\nu$$:=$$w$ is an assignment.
$[\![w]\!]m$ denotes that the value evaluating by the expression $w$ under the memory state $m$. And this value is assigned to the variable $\nu$.
Thus, $m'$$=$$m[\nu$$\mapsto$$[\![w]\!]m]$ denotes the new memory state where $m'(\nu)$$=$$[\![w]\!]m$ and $m'(y)$$=$$m(y)$ ($\forall y$$\in$$\mathcal{V}$$:y$$\neq$$\nu$) in the new state $s'$ for $\rho_{ass}$.
$\rho_{tcd}$ represents the boolean condition $w$ is evaluated by $true$ (i.e., $[\![w]\!]m$$=$$true$), and $\rho_{fcd}$ represents the boolean condition $w$ is evaluated by $flase$ (i.e., $[\![w]\!]m$$=$$flase$).
Neither $\rho_{tcd}$ nor $\rho_{fcd}$ updates the memory state. And they update the program locations to different program locations $l'$.

$\rho_{acq}$ represents thread $i$ obtains this mutex $\ell$ ($r[\ell$$\mapsto$$i]$) and updates the program location to $l'$ if $\ell$ is not held by any thread ($r(\ell)$$=$$0$).
However, if $\ell$ is held by another thread, thread $i$ is blocked by $\ell$, and current state $s$ cannot be updated.
$\rho_{rel}$ corresponds to $\langle unlock, \ell\rangle$.
Here, $r(\ell)$$=$$i$ means the mutex $\ell$ is held by thread $i$. If $r(\ell)$$=$$i$, thread $i$ releases this mutex $\ell$ by $r[\ell$$\mapsto$$0]$ and update the program location to $l'$.
$\rho_{sig}$ corresponds to $\langle signal, \gamma\rangle$. Thread $i$ can notify a thread $j$ belonging to $u(\gamma)$ ($\{j\}$$\in$$u(\gamma)$). Thus, thread $j$ will be notified by thread $i$ ($u[\gamma$$\mapsto$$u(\gamma)$$\setminus$$\{j\}$$\cup$$\{-j\}]$). And it updates the program location to $l'$. 
If the mutex $\ell$ is held by thread $i$ ($r(\ell)$$=$$i$) and thread $i$ is not waiting for $\gamma$ currently ($\{i\}$$\notin$$u(\gamma)$), $\rho_{wa_1}$ represents the action releases the mutex $\ell$ ($r[\ell$$\mapsto$$0]$), and thread $i$ is added to the current thread multiset waiting on condition variable $\gamma$ ($u[\gamma$$\mapsto$$u(\gamma)$$\cup$$\{i\}]$).
Then, $\rho_{}wa_2$ represents the thread $i$ is blocked until a thread $j$ ($\{-i\}$$\in$$u(\gamma)$) notifies by condition variable $\gamma$ . Thus, thread $i$ no long waits for a notification on $\gamma$ ($u[\gamma$$\mapsto$$u(\gamma)$$\setminus$$\{-i\}]$).
Finally, if the thread $i$ has been woken up by the other thread and $\ell$ is not held ($r(\ell)=0$), $\rho_{wa_3}$ represents the action acquires the mutex $\ell$ again ($r[\ell$$\mapsto$$i]$) and updates the program location to $i$.

\end{subappendices}
\end{document}